\documentclass[runningheads]{llncs}

\usepackage{graphicx}
\usepackage{balance}  
\usepackage{amssymb}

\usepackage{graphicx}
\usepackage{caption}
\usepackage{subcaption}
\usepackage{xspace}
\usepackage{alltt}
\usepackage{graphicx}
\usepackage{caption}
\usepackage{etex}

\usepackage{hyperref}

\usepackage{multirow}

\usepackage{xcolor}

\usepackage{tikz}
\usetikzlibrary{arrows,decorations,backgrounds}
\usetikzlibrary{shapes,calc,backgrounds,shadows,fit,patterns}
\usetikzlibrary{decorations.markings}
\usetikzlibrary{positioning,arrows,fit,snakes}
\usetikzlibrary{decorations.pathmorphing}
\usetikzlibrary{calc}
\usetikzlibrary{patterns}

\tikzset{
    axis break gap/.initial=1mm
}

\usepackage{amsmath,mathtools}
\usepackage{multicol}
\usepackage{times}
\normalfont 

\usepackage{pifont}
\usepackage{microtype}

\let\llncssubparagraph\subparagraph
\let\subparagraph\paragraph
\usepackage{titlesec}
\titlespacing{\section}{0pt}{7pt}{5pt}
\titlespacing{\subsection}{0pt}{6pt}{4pt}
\titlespacing{\subsubsection}{0pt}{5pt}{3pt}
\titlespacing{\paragraph}{0pt}{4pt}{2pt}
\let\subparagraph\llncssubparagraph
\usepackage{multirow}
\usepackage{hhline}
\usepackage{latexsym}
\usepackage{amssymb,amsmath}
\usepackage{url}
\usepackage{booktabs}
\usepackage{rotating}
\usepackage{paralist,enumerate}
\usepackage{stmaryrd}
\usepackage{verbatim}

\usepackage{fancyvrb}
\usepackage[inline]{enumitem}

\usepackage{graphicx}
\usepackage{caption}

\usepackage{fancybox}
\usepackage{listings}
\usepackage{fancyvrb}
\usepackage{txfonts}
\usepackage{alltt}

\usepackage{pgfplots} 

\def\soledge[#1,#2,#3,#4];{
        \draw[dashed,-latex,#4] (#1) -- #3 (#2);
}

\usepackage{stfloats}
\usepackage{float}
\usepackage{listings}

\newif\ifdraft
\drafttrue

\usepackage{todonotes}

\ifdraft
\usepackage[ulem=normalem]{mychanges}
\else
\usepackage[final,ulem=normalem]{mychanges}
\fi
\setauthormarkup[right]{{\scriptsize[#1]}}
\definechangesauthor{DC}{blue}
\definechangesauthor{DL}{red}
\definechangesauthor{RW}{green}
\definechangesauthor{GX}{purple}
\definechangesauthor{DH}{yellow}

\renewcommand{\vec}[1]{\textbf{#1}}

\newcommand{\myPar}[1]{\paragraph{\textbf{\emph{#1}}}}

\newcommand{\statoil}{{Statoil}\xspace}

\newcommand{\fd}{{FD}\xspace}
\newcommand{\vfd}{VFD\xspace}
\newcommand{\fds}{{FDs}\xspace}
\newcommand{\vfds}{VFDs\xspace}
\newcommand{\oce}{OCE\xspace}

\newcommand{\unf}{\operatorname{unf}}

\newcommand{\rdfssubclassof}{\texttt{\scriptsize rdfs:subClassOf}}
\newcommand{\owlobjectpropertydomain}{\texttt{\scriptsize rdfs:domain}}
\newcommand{\owlobjectpropertyrange}{\texttt{\scriptsize rdfs:range}}
\newcommand{\rdftype}{\texttt{\scriptsize rdf:type}}
\newcommand{\SELECT}{{\bfseries \scriptsize\textcolor{black}{SELECT}}}
\newcommand{\AS}{{\bfseries \scriptsize\textcolor{black}{AS}}}
\newcommand{\FROM}{{\bfseries \scriptsize\textcolor{black}{FROM}}}
\newcommand{\WHERE}{{\bfseries \scriptsize \textcolor{black}{WHERE}}}

\newcommand{\JOIN}{{\bfseries \scriptsize \textcolor{black}{JOIN}}}

\newcommand{\proj}[1]{\pi_{#1}}
\newcommand{\tra}{\boldsymbol{\tau}}

\newcommand{\var}{\textit{var}}
\newcommand{\cc}{\rule[-0.2ex]{0.7pt}{1.8ex}\rule{0.2em}{0pt}\rule[-0.2ex]{0.7pt}{1.8ex}}

\newcommand{\sANS}[2]{\llbracket #1\rrbracket_{#2}}

\newcommand{\isNull}{\textit{isNull}}
\newcommand{\rTwoTr}{{\sf tr}}
\definecolor{dark-gray}{gray}{0.3}
\newcommand{\owl}[1]{\textcolor{dark-gray}{\texttt{\scriptsize #1}}}
\newcommand{\owlind}[1]{\texttt{\scriptsize #1}}
\newcommand{\textowl}[1]{{\textcolor{dark-gray}{\texttt{\small #1}}}}
\newcommand{\SQL}[1]{\texttt{\scriptsize #1}}

\definecolor{midgrey}{rgb}{0.5,0.5,0.5}
\definecolor{darkred}{rgb}{0.7,0.1,0.1}

\definecolor{midgrey}{rgb}{0.5,0.5,0.5}
\definecolor{darkred}{rgb}{0.7,0.1,0.1}

\newcommand{\dl}[1]{\textcolor{red}{\textbf{!}}\todo[color=orange!30]{\footnotesize\textcolor{darkred}{\textbf{dl::}}\textcolor{blue}{\textbf{[}}#1\textcolor{blue}{\textbf{]}}}}

\newcommand{\davide}[1]{$[$\textcolor{darkred}{davide}:~~\emph{\textcolor{midgrey}{#1}}$]$}

\newcommand{\OWLQL}{\textrm{OWL\,2\,QL}\xspace}

\newcommand{\Null}{\textit{null}}
\newcommand{\sDOM}{C}
\newcommand{\dom}{\textit{dom}}

\newcommand{\tup}[1]{( #1 )}

\newcommand{\A}{\mathcal{A}}

\newcommand{\T}{\mathcal{T}}
\newcommand{\G}{\mathcal{G}}
\newcommand{\I}{\mathcal{I}}
\newcommand{\M}{\mathcal{M}}
\renewcommand{\O}{\mathcal{O}}
\renewcommand{\S}{\mathcal{S}}
\newcommand{\C}{\mathcal{C}}

\newcommand{\ontop}{\textsl{Ontop}\xspace}

\newcommand{\join}{\textsc{Join}}

\newcommand{\leftjoin}{\textsc{Opt}}
\newcommand{\union}{\textsc{Union}}
\newcommand{\filter}{\textsc{Filter}}

\newcommand{\nm}[2][]{\textit{#2}}

\newcommand{\myIRI}{\texttt{I}}
\newcommand{\myBNK}{\texttt{B}}
\newcommand{\myLIT}{\texttt{L}}
\newcommand{\myVAR}{\texttt{V}}

\newcommand{\ignore}[1]{}

\newcommand{\mysql}{\textsc{MySQL}\xspace}
\newcommand{\dbtwo}{\textsc{DB2}\xspace}
\newcommand{\postgres}{\textsc{PostgreSQL}\xspace}

\newcommand{\sparql}{\textsc{SPARQL}\xspace}
\newcommand{\sql}{\textsc{SQL}\xspace}

\newcommand{\tableRowBreak}{\\}

\newcommand{\tableColSpace}{@{\hspace{1.5em}}}

\newcommand{\tableColSmallSpace}{@{\hspace{1ex}}}

\newenvironment{theoremnum}[1]{\smallskip\noindent\textbf{Theorem~#1.}
  \hspace*{0.3em}\em}{\par\smallskip}
\newenvironment{lemmanum}[1]{\smallskip\noindent\textbf{Lemma~#1.}
  \hspace*{0.3em}\em}{\par\smallskip}
\newenvironment{propositionnum}[1]{\smallskip\noindent\textbf{Proposition~#1.}
  \hspace*{0.3em}\em}{\par\smallskip}

\newif\ifextendedversion
\extendedversiontrue

\begin{document}



\def\tit{OBDA Constraints for Effective Query Answering}

\ifextendedversion %
\title{\tit\\(Extended Version)} %
\else %
\title{\tit} %
\fi

\author{Dag Hovland$^1$ \and Davide Lanti$^2$ \and Martin Rezk$^2$
  \and Guohui Xiao$^2$}

\institute{
$^1$University of Oslo, Norway\\
$^2$Free University of Bozen-Bolzano, Italy
}
\maketitle

\begin{abstract}
  In Ontology Based Data Access (OBDA) users pose SPARQL queries over
  an ontology that lies on top of relational datasources. These
  queries are translated on-the-fly into SQL queries by OBDA systems.
  Standard SPARQL-to-SQL translation techniques in OBDA often produce
  SQL queries containing redundant joins and unions, even after a
  number of semantic and structural optimizations.  These redundancies
  are detrimental to the performance of query answering, especially in
  complex industrial OBDA scenarios with large enterprise databases.
  To address this issue, we introduce two novel notions of OBDA
  constraints and show how to exploit them for efficient query
  answering. We conduct an extensive set of experiments on large
  datasets using real world data and queries, showing that these
  techniques strongly improve the performance of query answering up to
  orders of magnitude.
\end{abstract}



\section{Introduction}\label{sec:intro}

In Ontology Based Data Access (OBDA)~\cite{PLCD*08}, the complexity of data storage is hidden by
a conceptual layer on top of an existing relational database (DB). 
Such a conceptual layer, realized by an ontology, provides a convenient vocabulary for user queries, and captures domain knowledge (e.g., hierarchies of concepts) that can be used to enrich 
query answers over incomplete data. 
%
The  ontology is connected to the relational database through a declarative specification given
in terms of mappings that relate each term in the ontology (each class and property) to a (SQL) view over the database.
The mappings and  the database define a \emph{(virtual)} RDF graph
that, 
together with the ontology, can be queried using the \emph{SPARQL query language}.

To answer a SPARQL query over the conceptual layer, a typical OBDA system translates it into an equivalent SQL query over the original database.
The translation procedure has two major stages: (1) \textit{rewriting} the input SPARQL query with respect to the ontology and (2) \textit{unfolding} the rewritten query with respect to the mappings.
A well-known theoretical result is that the size of the translation is worst-case exponential in the size of the
input query~\cite{KKPZ12b}. These worst-case scenarios are not only theoretical, but they also occur in real-world
applications, as shown in~\cite{LRXC15}, where some user SPARQL queries are translated into SQL queries containing thousands  of join and union operators. This is mainly due to~\emph{(i)} SPARQL
queries containing joins of ontological terms with rich hierarchies,
which lead to redundant unions~\cite{RoKZ13}; and \emph{(ii)}
reifications of n-ary relations in the database into triples over the
RDF data model, which lead to SQL translations containing several
(mostly redundant) self-joins.
How to reduce the impact of exponential blow-ups through optimization
techniques so as to make OBDA applicable to real-world scenarios is
one of the main open problems in current OBDA research.

The standard solutions to tackle this problem are based on
\emph{semantic and structural
  optimizations}~\cite{RoKZ13,journalMariano} originally from the
database area~\cite{Chakravarthy:1986:SQO:21004.21043}. Semantic
optimizations use explicit integrity constraints (such as primary and
foreign keys) to remove redundant joins and unions from the translated
queries. Structural optimizations are in charge of reshaping the
translations so as to take advantage of database indexes. 

The main problem addressed in this paper is that these optimizations
cannot exploit constraints that go beyond database dependencies, such
as domain constraints (e.g., people have only one age, except for
Chinese people who have two ages), or storage policies in the
organization (e.g., table {\small\texttt{married}} \emph{must} contain
all the married employees).  We address this problem by
proposing two novel classes of constraints that go beyond database
dependencies.  The first type of constraint, \emph{exact predicate},
intuitively describes classes and properties whose elements can be
retrieved without the help of the ontology.  The second type of
constraint, \emph{virtual functional dependency} (\vfd), intuitively
describes a functional dependency over the virtual RDF graph exposed
by the ontology, the mappings, and the database.  These notions are
used to enrich the OBDA specification so as to allow the OBDA system
to identify and prune redundancies from the translated queries.  To
help the design of enriched specifications, we provide tools that
detect the satisfied constraints within a given OBDA
instance. 
%
 We extend the OBDA system \ontop so as to exploit the enriched specification, 
 and evaluate it in both a large-scale industrial setting
 provided by the petroleum company \statoil, and in an ad-hoc artificial and scalable benchmark with different commercial and free
relational database engines as back-ends.
Both sets of experiments reveal a drastic reduction on the size
of translated queries, which in some cases is reduced by orders of
magnitudes. This allows for a major performance improvement of query
answering.
%

The rest of the paper is structured as follows:
Preliminaries are provided in Section~\ref{sec:background}. 
In Section~\ref{sec:sparql-qa} we describe how state-of-the-art OBDA systems work, and highlight the problems with the current optimization techniques.
In Section~\ref{sec:tuningGeneral} we formally introduce our novel OBDA constraints, and show how they can be used to optimize translated queries.
In Section~\ref{sec:experiments} we provide an evaluation of the impact of the proposed optimization techniques on the performance
of query answering.
In Section~\ref{sec:related} we briefly survey other related works.
Section~\ref{sec:conclusions} concludes the paper.
\ifextendedversion %
The omitted proofs and extended experiments with Wisconsin benchmark
can be found in the appendix. %
\else %
Omitted proofs and extended experiments with Wisconsin benchmark can
be found in the extended version of this paper~\cite{techreport}. %
\fi



\section{Preliminaries}\label{sec:background}
We assume the reader to be familiar with relational algebra and SQL
queries, as well as with ontology languages and in particular with the
OWL~2~QL\footnote{\url{http://www.w3.org/TR/owl2-overview/}}
profile. 
%
To simplify the notation we express OWL~2 QL axioms by their
description logic
counterpart~\textit{DL-Lite}$_\mathcal{R}$~\cite{CDLLR07}.
%
Notation-wise, we will denote
tuples with the bold faces; e.g., \textbf{x} is a tuple.

\paragraph{\textbf{Ontology and RDF Graphs.}}
The building block of an ontology is a \emph{vocabulary} $(N_C, N_R)$,
where $N_C, N_R$ are respectively countably infinite disjoint sets of
\emph{class names} and (object or datatype) \emph{property names}. A \emph{predicate} is either a
class name or a property name. An \emph{ontology} is a finite set of
axioms constructed out a vocabulary, and it describes a domain of
interest.
 These axioms of an ontology can be 
serialized into a concrete syntax. In the following we use the
\emph{Turtle syntax} 
 for readability.

\begin{example}\label{ex:ontology}
  
  The ontology from \statoil captures the domain knowledge related to oil extraction activities.
  Relevant axioms for our examples are:\\
  \begin{tabular}{l \tableColSpace l}
    \midrule
    \textowl{$\owl{:isInWell}$ $\owlobjectpropertydomain{}$ $\owl{:Wellbore}$} & \textowl{$\owl{:isInWell}$ $\owlobjectpropertyrange{}$ $\owl{:Well}$}\\
    \textowl{$\owl{:hasInterval}$ $\owlobjectpropertydomain{}$ $\owl{:Wellbore}$} & \textowl{$\owl{:hasInterval}$ $\owlobjectpropertyrange{}$ $\owl{:WellboreInterval}$}\\
    \textowl{$\owl{:completionDate}$ $\owlobjectpropertydomain{}$ $\owl{:Wellbore}$}\\      
    \textowl{$\owl{:ProdWellbore}$ $\rdfssubclassof{}$ $\owl{:DevelopWellbore}$} &  \textowl{$\owl{:DevelopWellbore}$ $\rdfssubclassof{}$ $\owl{:Wellbore}$}\\
    \midrule
  \end{tabular}

  \noindent The first five axioms specify domains and ranges of the
  properties \owl{:isInWell}, \owl{:hasInterval}, and
  \owl{:completionDate}.  The last two  state the hierarchy
  between different wellbore\footnote{A wellbore is a three-dimensional representation of a hole in the ground.} classes.
\end{example}

Given a countably infinite set $N_I$ of \emph{individual names}
disjoint from $N_C$ and $N_R$, an \emph{assertion} is an expression of
the form $A(i)$ or $P(i_1, i_2)$, where $i, i_1, i_2 \in N_I$,
$A \in N_C, P \in N_R$. An OWL 2 QL \emph{knowledge base} (KB) is a
pair $(\T, \A)$ where $\T$ is an OWL 2 QL ontology and $\A$ is a set
of assertions (also called ABox). Semantics for entailment of assertions
 ($\models$) in OWL 2 QL KBs is
given through Tarski-style interpretations in the usual
way~\cite{BCMNP07}.
Given a KB $(\T, \A)$, the \emph{saturation of $\A$ with respect to $\T$} is the set of assertions
$\A_\T = \{A(s) \mid \tup{\T,\A}\models A(s)\} \cup \{P(s,o) \mid \tup{\T,\A}\models P(s,o)\}$. In the following, it is convenient to view assertions $A(s)$ and $P(s,o)$
as the \emph{RDF triples} $(s,\rdftype,A)$ and $(s,P,o)$, respectively
.
Hence, we view a set of assertions also as an RDF graph $\G^{\A}$ defined as 
$\G^{\A} = \{ (s,\rdftype,A) \mid A(s) \in \A \} \cup \{(s,P,o) \mid P(s,o) \in \A\}$. Moreover, the \emph{saturated RDF graph} $\G^{\tup{\T,\A}}$ associated to a knowledge base $\tup{\T,\A}$ consists of the set
of triples entailed by $\tup{\T,\A}$, i.e. $\G^{\tup{\T,\A}} = \G^{\A_\T}$.
           
\paragraph{\textbf{OBDA and Mappings.}}

Given a vocabulary $(N_C, N_R)$ and a database schema $\Sigma$, a \emph{mapping} is an expression of the form  
$A(f_1(\textbf{x}_1)) \leftarrow sql(\textbf{y})$ or $P(f_1(\textbf{x}_1), f_2(\textbf{x}_2)) \leftarrow sql(\textbf{y})$, 
where $A \in N_C$, $P \in N_R$, $f_1, f_2$ are function symbols, $\textbf{x}_i \subseteq \textbf{y}$, for $i=1, 2$, and $sql(\textbf{y})$ is an SQL query in $\Sigma$ having output attributes $\textbf{y}$.
Given $Q$ in $N_C \cup N_R$, \emph{a mapping $m$ is defining $Q$} if $Q$ is on the left hand side of $m$.

Given an SQL query $q$ and a DB instance $D$, $q^D$ denotes the set of answers to $q$ over $D$. 
Given a database instance $D$, and a set of mappings $\M$, we define the \emph{virtual
assertions set} $\A_{\M,D}$ as follows:
\[
  \A_{\M,D}=
  \begin{array}[t]{@{}l}
    \{A(f(\textbf{o})) \mid \textbf{o} \in \pi_{\textbf{x}}(sql(\textbf{y}))^D \text{ and } A(f(\textbf{x})) \leftarrow sql(\textbf{y})  
    \text{ in } \M \} \quad\cup{}\\
    \{P(f(\textbf{o}),g(\textbf{o'})) \mid (\textbf{o},\textbf{o'})\in \pi_{\textbf{x}_1,\textbf{x}_2}(sql(\textbf{y}))^D \text{ and }
     ~P(f(\textbf{x}_1), g(\textbf{x}_2)) \leftarrow sql(\textbf{y}) \text{ in } \M \}
  \end{array}
  \]
  In the Turtle syntax for mappings, we use
  \emph{templates}--strings with placeholders--for specifying the functions (like $f$ and $g$
  above) that map database values into URIs and literals.
For instance, the string
  \owl{\small <http://statoil.com/\{id\}>} is a URI template where
  ``\owl{\small id}'' is an attribute; when \owl{\small id} is
  instantiated as ``\owl{\small 1}'', it generates the URI \owl{\small
    <http://statoil.com/1>}.


An \emph{OBDA specification} is a triple $\mathcal{S} = (\T, \M, \Sigma)$ where $\T$
is  an ontology, $\Sigma$ is a database schema with key dependencies, and $\M$ is a set of mappings between $\T$ and $\Sigma$.
Given an OBDA specification $\S$ and a database instance $D$, we call the pair $(\S, D)$ an \emph{OBDA instance}.
Given an OBDA instance $\O = ((\T, \M, \Sigma), D)$, the \emph{virtual RDF graph exposed by $\O$} is
the RDF graph $\G^{\A_{\M,D}}$; the \emph{saturated virtual RDF graph $\G^\O$ exposed by $\O$} is the RDF graph $\G^{\tup{\T, \A_{\M,D}}}$.

\begin{example}\label{fig:mappingsEPDS}
  The mappings for the classes and properties introduced in Example~\ref{ex:ontology} are:

    \begin{lstlisting}[frame=tb]
$\owlind{:Wellbore-\{wellbore\_s\}} \ \owl{rdf:type} \ \owl{:Wellbore}$
$\leftarrow \SQL{\SELECT{} wellbore\_s}$ $\SQL{\FROM{}  wellbore}$ $\SQL{\WHERE{} wellbore.r\_existence\_kd\_nm = 'actual' }$

$\owlind{:Wellbore-\{wellbore\_s\}} \ \owl{:isInWell} \ \owlind{:Well-\{well\_s\}}$
$\leftarrow \SQL{\SELECT{} well\_s, wellbore\_s}$ $\SQL{\FROM{}  wellbore}$ $\SQL{\WHERE{} wellbore.r\_existence\_kd\_nm = 'actual' }$

$\owlind{:Wellbore-\{wellbore\_s\}} \ \owl{:hasInterval}$ $\owlind{:WellboreInterval-\{wellbore\_intv\_s\}}$
$\leftarrow \SQL{\SELECT{} wellbore\_s, wellbore\_intv\_s}$ $\SQL{\FROM{}  wellbore\_interval}$

$\owlind{:Wellbore-\{wellbore\_s\}} \ \owl{:completionDate} $ $\owlind{`\{year\}-\{month\}-\{day\}'\^{}\^{}xsd:date}$
$\leftarrow \SQL{\SELECT{} wellbore\_s, year, month, day}$ $\SQL{\FROM{}  wellbore}$ $\SQL{\WHERE{} wellbore.r\_existence\_kd\_nm = 'actual' }$
      
$\owlind{:Wellbore-\{wellbore\_s\}} \ \owl{rdf:type} \ \owl{:ProdWellbore}$
$\leftarrow \SQL{\SELECT{} w.wellbore\_s AS wellbore\_s}$ $\SQL{\FROM{} wellbore w, facility\_clsn }$ $\SQL{\WHERE{}} $ complex-expression
\end{lstlisting}
\end{example}

\paragraph{\textbf{Query Answering in OWL 2 QL KBs.}}

A \emph{conjunctive query} $q(\textbf{x})$ is a first order formula of
the form $\exists \textbf{y}.\ \varphi(\textbf{x}, \textbf{y})$, where
$\varphi(\textbf{x}, \textbf{y})$ is a conjunction of equalities and
atoms of the form $A(t)$, $P(t_1, t_2)$ (where
$A \in N_C, P \in N_R$), and each $t, t_1, t_2$ is either a
\emph{term} or an individual variable in $\textbf{x},
\textbf{y}$. 
Given a conjunctive query $q(\textbf{x})$ and a knowledge
base $\mathcal{K} := (\T, \A)$, a tuple
$\textbf{i} \in N_I^{|\textbf{x}|}$ is a \emph{certain answer} to
$q(\textbf{x})$ iff $\mathcal{K} \models q(\textbf{i})$. 
The task of query answering in OWL~2~QL (\textit{DL-Lite}$_\mathcal{R}$) can be addressed by query
rewriting techniques~\cite{CDLLR07}.
For an OWL~2~QL ontology $\T$,
 a conjunctive query $q$ can be rewritten to a union $q_r$ of
conjunctive queries  such that for each assertion set $\A$ and
each tuple of individuals $\textbf{i} \in N_I^{|\textbf{x}|}$, it
holds
$\mathcal{(\T,\A)} \models q(\textbf{i}) \Leftrightarrow \A \models
q_r(\textbf{i})$. Many rewriting techniques have been proposed in the
literature~\cite{KiKZ12,RoAl10,BOSX13}. 

SPARQL~\cite{W3Crec-SPARQL-1.1-entailment} is a W3C standard language designed to query RDF graphs.
Its vocabulary contains four pairwise disjoint and countably infinite sets of symbols: \myIRI{} for \emph{IRIs},  \myBNK{} for \emph{blank nodes}, \myLIT{} for \emph{RDF literals}, and \myVAR{} for \emph{variables}. The elements of $\sDOM = \myIRI \cup \myBNK \cup \myLIT$ are called \emph{RDF terms}.
A \emph{triple pattern} is an element of $(\sDOM\cup\myVAR)\times \myIRI \times (\sDOM\cup\myVAR)$. A \emph{basic graph pattern} (\emph{BGP}) is a finite set of joins of triple patterns. BGPs can be combined
using the SPARQL operators join, optional, filter, projection, etc.

\begin{example}
\label{ex:sparql}
  The following \sparql query, containing a BGP with three triple patterns, returns all the wellbores, their completion dates, and the well where they are contained.
  \begin{lstlisting}[basicstyle=\ttfamily\scriptsize,mathescape,frame=none,frame=tb]
$\SQL{\SELECT{} * \WHERE{} \{?wlb \owl{rdf:type} \owl{:Wellbore}.}$ $ \texttt{?wlb} \owl{\texttt{:completionDate}} \owlind{~?cmpl}.\  \texttt{?wlb} \owl{\texttt{:isInWell}} \owlind{~?w} . \}$
  \end{lstlisting}
\end{example}
\vspace{-.2cm}

\noindent{}To ease the presentation of the technical development, in
the rest of this paper we adopt the OWL~2 QL entailment regime for
SPARQL query answering~\cite{KRRXZ14}, but disallow complex
class/property expressions in the query.
Intuitively this restriction states that
each BGP can be seen as a conjunctive query without
existentially quantified variables.
%
Under this restricted OWL~2 QL entailment regime, the task of
answering a \sparql query $q$ over a knowledge base $\tup{\T,\A}$ can
be reduced to answering $q$ over the saturated graph
$\G^{\tup{\T,\A}}$ under the simple entailment regime.
This restriction can be lifted with the help of a standard query
rewriting step~\cite{KRRXZ14}.



\section{SPARQL Query Answering in OBDA}
\label{sec:sparql-qa}
In this section we describe the typical steps that an OBDA system performs to answer SPARQL queries 
and discuss the performance challenges. To do so, we pick
the representative state-of-the-art OBDA system \ontop and discuss its functioning in detail.

During its start-up, \ontop classifies the ontology, ``compiles'' the ontology
into the mappings generating the so-called $\T$-mappings~\cite{RoKZ13}, and removes redundant 
mappings by using inclusion dependencies (e.g., foreign keys) contained in the database schema.
Intuitively, $\T$-mappings expose a saturated RDF graph. Formally, given a basic OBDA specification $\mathcal{S} = (\T, \M, \Sigma)$, the mappings $\M_\T$ are $\T$-mappings for $\S$ if, for every OBDA instance $\mathcal{O} = (\mathcal{S}, D)$, $\G^\O = \G^{\tup{\A_{\M_\T},D}}$.

\begin{example}
  \label{ex:statoil-tmapping}
  The $\T$-mappings for our running example are those in Example~\ref{fig:mappingsEPDS} plus
  
    \begin{lstlisting}[basicstyle=\ttfamily\scriptsize,mathescape,frame=none,frame=tb]
$\owlind{:Wellbore-\{wellbore\_s\}} \ \owl{rdf:type} \ \owlind{:Wellbore}$
$\leftarrow \SQL{\SELECT{} wellbore\_s \FROM{}  wellbore}$ $\SQL{\WHERE{} wellbore.r\_existence\_kd\_nm = 'actual' }$

$\owlind{:Wellbore-\{wellbore\_s\}} \  \owl{rdf:type} \ \owlind{:Wellbore}$
$\leftarrow \SQL{\SELECT{} wellbore\_s, wellbore\_intv\_s \FROM{}  wellbore\_interval}$

$\owlind{:Wellbore-\{wellbore\_s\}} \ \owl{rdf:type} \ \owlind{:Wellbore}$
$\leftarrow \SQL{\SELECT{} w.wellbore\_s}$ $\SQL{\FROM{} wellbore w, facility\_clsn }$    $\SQL{\WHERE{} } $ ... complex-expression
    \end{lstlisting}

    The new mappings are derived from the domain of the properties \texttt{\small :isInWell, :completionDate}, and because  
    \texttt{\small :ProdWellbore} is a sub-class of \texttt{\small :Wellbore}.

\end{example}
%

After the start-up, in the query answering
stage, \ontop translates the input SPARQL query
into an SQL query, evaluates it, and returns the answers to the end-user.
We divide this stage in five phases:
\begin{inparaenum}[\it(a)]
\item\label{p:rew} the SPARQL query is \emph{rewritten} using the tree-witness rewriting algorithm;
\item\label{phase:unf} the rewritten SPARQL query is \emph{unfolded} into an SQL query using $\T$-mappings;
\item\label{phase:opt} the resulting SQL query is optimized;
\item\label{p:ex} the optimized SQL query is  executed by the database engine;
\item\label{p:res} the SQL result is translated into the answer to the original  SPARQL query.
\end{inparaenum}
For the sake of simplicity, we disregard phase~\emph{(\ref{p:rew})} since it goes out of the
scope of this paper (cf.~\cite{GKKP*14}), and phases~\emph{(\ref{p:ex})} and~\emph{(\ref{p:res})} because they are straightforward.
In the following we elaborate on phases~\emph{(\ref{phase:unf})} and~\emph{(\ref{phase:opt})}.


\paragraph{\textbf{From SPARQL to SQL.}}
In phase~\emph{(\ref{phase:unf})} the rewritten SPARQL query is unfolded into an SQL query using $\T$-mappings. The rewritten query
is first transformed into a tree representation of its SPARQL algebra expression. 
The algorithm starts by replacing each leaf of the tree, that is, a
triple pattern of the form $(s,p,o)$, with the union of the SQL
queries defining $p$ in the $\T$-mapping. Such \sql queries are obtained as follows: 
given a triple pattern $p = \owlind{?x } \owl{rdf:type } \owl{:A}$, and a
mapping $m = \owl{\normalsize :A}(f(\mathbf{y'})) \leftarrow sql(\textbf{y})$, 
 the \emph{SQL unfolding
  $\unf(p,m)$ of $p$ by $m$} is the SQL query
$\SQL{\SELECT{} } \tau(f({\mathbf{y'}})) \SQL{ \AS{} } \owlind{x} \SQL{ \FROM{} } sql(\textbf{y})$, 
    where $\tau$ is an SQL function filling the placeholders in $f$ with values in $\mathbf{y'}$. 
 We denote the sub-expression ``$\SQL{\SELECT{} 
    } \tau(f(\mathbf{y'})) \SQL{ \AS{} } \owlind{x}$'' by $\pi_{x/f(\mathbf{y'})}$.  The
    notions of ``$\unf$'' and ``$\pi$'' are defined similarly for
    properties.

\begin{example}
  Consider the triple pattern $p = \owlind{?wlb } \owl{:completionDate } \owlind{?d}$, and the fourth mapping $m$ from Example~\ref{fig:mappingsEPDS}. Then
  the SQL unfolding $\unf(p,m)$ is the SQL query
  \begin{lstlisting}[basicstyle=\ttfamily\scriptsize,mathescape,frame=none,frame=tb]
$\SQL{\SELECT{} CONCAT(":Wellbore-",well\_s) \AS{} wlb,}$$\SQL{CONCAT("`",year,"-",month,"-", day,"'\^{}\^{}xsd:date") \AS{} d }$ 
$\SQL{\FROM{}  wellbore}$ $\SQL{\WHERE{} wellbore.r\_existence\_kd\_nm = 'actual' }$
    \end{lstlisting}
\end{example}
\vspace{-.2cm}

Given a triple pattern $p$ and a set of mappings $\M$, the
\emph{SQL unfolding $\unf(p, \M)$ of $p$ by $\M$} is the SQL union
$\cup_{m\in\M} \{\unf(p,m) \mid \unf(p,m) \textit{ is defined} \}$.

Once the leaves are processed, the algorithm processes the upper levels in the
tree, where the SPARQL operators  are
translated into the corresponding SQL operators (Project, InnerJoin, LeftJoin,
Union, and Filter). 
Once the root is translated the process terminates 
and the resulting SQL expression is returned.

    \begin{example}\label{ex:trivRed}
      The unfolded SQL query for the SPARQL query in
      Example~\ref{ex:sparql} and $\T$-mappings in Example~\ref{ex:statoil-tmapping} has the following shape:
      \vspace{-0.5mm}
      \begin{align*}
        & (\pi_{wlb/\Box}sql_{\owl{{:}Wellbore}} \cup \pi_{wlb/\Box}sql_{\owl{{:}ProdWellbore}} \cup \pi_{wlb/\Box}sql_{\owl{:hasInterval}})\\
        &  \quad\Join (\pi_{wlb/\Box, cmp/\Diamond}sql_{\owl{{:}completionDate}}) \Join (\pi_{wlb/\Box, w/\circ}sql_{\owl{{:}isInWell}})       
      \end{align*}
\vspace{-0.5mm}
      where $\small \Box = \owl{:Wellbore-\{wellbore\_s\}}$, $\small \Diamond = \owl{`\{year\}-\{month\}-\{day\}'\^{}\^{}xsd:date}$, $\small \circ = \owl{:Well-\{well\_s\}}$, and $sql_P$ is the SQL query in the mapping defining the class/property $P$. 
    \end{example}

\ignore{
\begin{example}\label{ex:simple-statoil}
  Consider the second query from Example \ref{ex:intro}. 
  In SPARQL this query can be written:
\begin{lstlisting}[basicstyle=\ttfamily\scriptsize,mathescape,frame=none,frame=tb]
$\SQL{\SELECT{} * \WHERE{} \{?w \owl{rdf:type} \owl{:Wellbore};}$ $ \owl{\texttt{:completionDate}} \owlind{~?cmpl}; \owl{\texttt{:isInWell}} \owlind{~?w} . \}$
  \end{lstlisting}
The above query corresponds to the a SPARQL algebra tree where  the leaves are the triple patterns,
and the inner nodes are the SPARQL operators (such as Join).

 \end{example}
}


\myPar{Optimizing the generated SQL queries.}\label{sec:optimize:SQL}
At this point, the unfolded SQL queries are merely of theoretical value as they would not be efficiently executable by any database system. A problem comes from the fact that they contain joins over the results of built-in database functions, which are expensive to evaluate. Another problem is that the unfoldings are usually verbose, often containing thousands of unions and join operators. Structural and semantic optimizations are in charge of dealing with these two problems.

\medskip
\noindent\emph{Structural Optimizations.}~
To ease the presentation, we assume the queries to contain only one BGP.
Extending to the general case is straightforward.  
An SQL unfolding of a BGP has the shape of a join of unions $Q = Q_1 \Join Q_2 \ldots \Join Q_n$, 
where each $Q_i$ is a union of sub-queries. 
The first step is to remove duplicate sub-queries in each $Q_i$.
In the second step, $Q$ is transformed into a union of joins.
In the third step, all joins of the kind $\pi_{x/f}sql_1(\textbf{z}) \Join \pi_{x/g}sql_2(\textbf{w})$ where $f \neq g$ are removed because they do not produce any answer.
In the fourth step, the occurrences of the SQL function $\pi$ for creating URIs are pushed to the root of the query tree so as to obtain
efficient queries where the joins
are over database values rather than over URIs.
Finally, duplicates in the union are removed.

\medskip \noindent\emph{Semantic Optimizations.}
SQL queries are semantically analyzed with the
goal of transforming them into a more efficient form. The analyses are based on database integrity constraints (precisely, primary and foreign keys)
explicitly defined in the database schema. These constraints are used to identify and remove redundant self-joins and unions from the unfolded SQL query.
%

\paragraph{\textbf{How Optimized are Optimized Queries?}}

There are real-world cases where the optimizations discussed above
are not enough to mitigate the exponential explosion caused by the unfolding.
As a result, the unfolded SQL queries
cannot be efficiently handled by DB engines~\cite{LRXC15}. However, the same
queries can usually be manually formulated in a succint way by database managers. A reason for this is that database dependencies cannot model certain domain constraints or storage policies that are
available to the database manager but not to the OBDA system. The next
example, inspired by the \statoil use case explained in Section \ref{sec:experiments}, illustrates this issue.

\begin{example}\label{ex:intro}

  The data stored at \statoil has certain properties that derive from domain constraints or storage policies. Consider a modified version of the query defining the class \textowl{:Wellbore} where all the attributes are projected out. According to storage policies for the database table \SQL{\small wellbore}, the result of the evaluation of this query against any database instance must satisfy the following constraints:
  \begin{enumerate*}[ label=\it(\roman*)]
  \item\label{item1} it must contain all the wellbores\footnote{i.e., individuals in the class \textowl{:Wellbore}} in the ontology (modulo templates);
  \item\label{item2} every tuple in the result must contain the information about name, date, and well (no nulls);
  \item\label{item3} for each wellbore in the result, there is exactly one date/well
    that is tagged as `actual'.
  \end{enumerate*}
  
  \paragraph{Query with Redundant Unions.} 
  Consider the SPARQL query retrieving all the wellbores, namely $\SQL{\SELECT{} * \WHERE{} \{?wlb \owl{rdf:type} \owl{:Wellbore}.\}}$.
  By ontological reasoning, the query will retrieve also the wellbores that can be inferred from the subclasses of \textowl{:Wellbore}
  and from the properties where \textowl{:Wellbore} is the domain or range.
  Thus, after unfolding and optimizations, the resulting SQL query has the structure $\pi_{wlb/\Box}(sql_1)$, with $sql_1 = (sql_{\owl{{:}Wellbore}} \cup sql_{\owl{{:}ProdWellbore}} \cup \pi_{\#}sql_{\owl{:hasInterval}})$,     
  where $\Box$ \texttt{\small= :Wellbore-\{wellbore\_s\}}, and \texttt{\small \# = wellbore\_s}.
  However, all the answers returned by $sql_1$ are also returned by the query $sql_{\owl{:Wellbore}}$ alone,
  when these two queries are evaluated on a data instance satisfying item~\ref{item1}. 

  \paragraph{Query with Redundant Joins.} 
  For the SPARQL query in~Example~\ref{ex:sparql}, the unfolded and optimized SQL translation is of the form $\pi_{wlb/\Box, cmp/\Diamond, w/\circ}(sql_2)$ with  $sql_2 = sql_1 \Join sql_{\owl{:completionDate}} \Join sql_{\owl{:isInWell}}$.
  Observe that the answers from $sql_2$ could also be retrieved from a projection and a selection over \texttt{\small wellbore}.
  This is because  $sql_1$ could be simplified  to  $sql_{\owl{:Wellbore}}$ and items~\ref{item2} and~\ref{item3}. 
  The problem we highlight here is that this ``optimized'' SQL query contains two redundant joins if storage policies and domain constraints are taken into account.  
\end{example}

It is important to remark that the constraints in the previous example cannot be expressed through schema dependencies like foreign or primary keys (because these constraints are defined over the output relations of SQL queries in the mappings, rather than over database relations\footnote{Materializing the SQL in the mappings is not an option, since the schema is fixed.}). Therefore, current state-of-the-art optimizations applied in OBDA cannot exploit this information.

\section{OBDA Constraints}\label{sec:tuningGeneral}

We now formalize two properties over an OBDA instance: 
\emph{exact predicates} and \emph{virtual functional dependencies}. 
We will then enrich the OBDA specification with a constraints component, 
stating that all the instances for the specification display such properties.
We show how this additional constraint component can be used to identify and remove redundant unions and joins from the unfolded queries.

From now on, let $\O = (\S, D)$ be an OBDA instance of a specification $\S = (\T, \M, \Sigma)$.


\subsection{Exact Predicates in an OBDA Instance}
\label{sec:exact-mappings-obda}

In real world scenarios it often happens that axioms in the ontology do not enrich the answers to queries. Often this is due to storage policies not available to the OBDA system.
%
This fact leads to redundant unions in the generated SQL, as shown in Example~\ref{ex:intro}.
In this section we show how certain properties defined on the mappings and the predicates, ideally deriving from such constraints, can be used to reduce the number of redundant unions
in the generated SQL queries for a given OBDA instance.


\begin{definition}[Exact Mapping]\label{dfn:exactmapping}
  Let $\M'$ be a set of
  mappings defining a predicate $A$.  We say that $\M'$ is exact for
  $A$ in $\O$ if 
  $\O \models A(\vec{a}) \text{ if and only if } ((\emptyset,\M', \Sigma),D)
  \models A(\vec{a}). $
\end{definition}



In practice it is often the case that the mappings for a particular
predicate declared in the OBDA specification are already exact. This leads
us to the  next definition.

\begin{definition}[Exact Predicate]
  A predicate $A$ is exact in $\O$ 
  if the set of all the mappings in $\M$ defining $A$ are exact for $A$
  in $\O$.
\end{definition} 

 Recall that \ontop adds new mappings to the initial set of mappings through the $\T$-mapping technique. For exact predicates, this can be avoided while producing the same saturated virtual RDF graph. 
Fewer mappings lead to unfoldings with less unions.

\begin{proposition} \label{prop:exact} 
  Let $\M'$ be exact for the predicate $A$ in $\T$.  
  Let $\M_\T'$ be the result of replacing all the mappings 
  defining $A$ in $\M_\T$ by $\M'$. 
	Then $\G^\O=\G^{(( \emptyset,\M_\T', \Sigma ), D)}$.
\end{proposition}

\begin{example}
The $\T$-mappings for \textowl{:Wellbore} consist of four mappings (see Example~\ref{ex:statoil-tmapping}). 
However, \textowl{:Wellbore} is an exact class (Example~\ref{ex:intro}). 
Therefore we can drop the three  $\T$-mappings for 
\textowl{:Wellbore} inferred from the ontology,
and leave only its  original  mapping.
\end{example}

\ignore{
\begin{corollary}
Given a SPARQL query $Q$, the answer of $Q$ over $(\T,\M,D)$ 
is the same as the answer of $Q$ over $(\emptyset, \M_\T', D)$, 
where $\M_\T'$ is as the mapping obtained in Theorem~\ref{prop:exact}.
\end{corollary}
}



\subsection{Functional Dependencies in an OBDA instance}\label{s:fds-in-obda}


Recall that in database theory 
 a functional dependency
(abbr. \fd) is an expression of the form
$\textbf{x}\rightarrow \textbf{y}$, read \emph{\textbf{x} functionally
  determines \textbf{y}}, where $\textbf{x}$ and $\textbf{y}$ are
tuples of attributes. We say that
\emph{$\textbf{x}\rightarrow\textbf{y}$ is over an attributes set $R$}
if $\textbf{x} \subseteq R$ and $\textbf{y}\subseteq R$. Finally,
\emph{$\textbf{x}\rightarrow \textbf{y}$ is satisfied by a relation
  $I$ on $R$} if $\textbf{x}\rightarrow \textbf{y}$ is over $R$ and
for all tuples $\vec{u},\vec{v}\in I$, if the value
$\vec{u}[\textbf{x}]$ of $\textbf{x}$ in $\vec{u}$ is equal to the
value $\vec{v}[\textbf{x}]$ of $\textbf{x}$ in $\vec{v}$, then
$\vec{u}[\textbf{y}]=\vec{v}[\textbf{y}]$. Whenever $R$ is clear from
the context, we simply say that $\textbf{x}\rightarrow \textbf{y}$ is
satisfied in $\I$.

A \emph{virtual functional dependency} intuitively describes a
functional dependency on a saturated virtual RDF graph.  We identify
two types of virtual functional dependencies:
\begin{compactitem}
\item \emph{Branching \vfd}: This dependency describes the relation
  between an object and a set of functional properties providing
  information about this object.  Intuitively, it corresponds to a
  ``star'' of ``functional-like''\footnote{A property which is
    functional when restricting its domain/range to individuals
    generated from a single template.} properties in the virtual RDF
  graph.  For instance, given a person, the properties describing its
  (unique) gender, national id, biological mother, etc. are a
  branching \vfd.

\item \emph{Path \vfd}: This dependency describes the case when, from
  a given individual and a list of properties, there is at most one
  path that can be followed using the properties in the list.
For instance, \emph{x} works in a single department \emph{y}, and \emph{y} has a single manager 
\emph{w}, and \emph{w} works for a single company \emph{z}. 
\end{compactitem}
We use these notions to identify those cases where a SPARQL join of
properties translates into a redundant SQL join.  \ignore{ Recall that
  during the optimization phase SQL joins like
  $\pi_{x/f}(sql_1) \Join \pi_{x/g}(sql_2)$ where $f \neq g$
  will be eliminated.  The SPARQL-to-SQL translation produces a join
  of two SQL queries from a SPARQL join of two properties \emph{only
    if} the mappings defining these properties share a template.
}

\begin{definition}[Virtual Functional Dependency]\label{def:nfd}
  Let $t$ be a template,
  and $S_t$ be the set of individuals in $G^{\O}$ generated from
  $t$. Let $P,P_1,\dots, P_n$ be properties in $\T$. Then
\begin{compactitem}
\item A \emph{branching \vfd} is an expression of the form
  $t \mapsto^b P_1 \cdots P_n$.  A \vfd \emph{$t \mapsto^b P$ is
    satisfied in $\O$} if for each element $s\in S_t$, there are no
  $o\neq o' \text{ in } G^{\O}$ such that
  $\{(s,P,o), (s,P,o')\}\subseteq G^{\O}$. A \vfd
  $t \mapsto^b P_1 \cdots P_n$ is satisfied in $\O$ if
  $t \mapsto^b P_i $ is satisfied in $\O$ for each
  $i \in \{1,\ldots,n\}$.

\item A \emph{path \vfd} is an expression of the form
  $t \mapsto^p{P_1 \cdots P_n} $.  A \vfd
  \emph{$t \mapsto^p{P_1 \cdots P_n} $ is satisfied in $\O$} if for
  each $s\in S_t$
  there is at most one list of nodes $(o_1,\dots, o_n)$ in
  $G^\mathcal{O}$ such that $ \{ (s,P_1,o_1), \dots,$
  $ (o_{n-1}, P_n, o_n) \}$ $ \subseteq G^\mathcal{O}$.
\end{compactitem}
\end{definition}
The next example shows, similarly as in~\cite{Weddell}, that 
general path \vfds cannot be expressed as a combination of path
\vfds of length 1.
\begin{example}
  Let  $\G^\O = \{(s, P_1, o_1), (o_1, P_2, o_2), (s, P_1, o_1')\}$,
  and $t$ a template such that $S_t = \{s\}$. Then,
  $t\mapsto^p{P_1P_2}$ is clearly satisfied in $\O$. However, $t\mapsto^p{P_1}$ is not.
\end{example}
A property $P$ might not be functional, 
but still $t\mapsto^b P$ might be satisfied in $\O$ for some $t$.
\begin{example}\label{ex:functProp-noVFD}
  Let 
  $\G^\O = \{(s, P, o_1), (s, P, o_2), (s', P, o_3)\}$, and $t$ a template such that $S_t = \{s'\}$. Then, the \vfd
  $t\mapsto^p{P}$ is satisfied in $\O$, but $P$ is not functional.
\end{example}
A functional dependency satisfied in the virtual RDF graph might not
correspond to a functional dependency over the database relations. We show this with 
an example:

\begin{example}\label{ex:opti1}
	Consider the following instance of the view \textowl{wellbore}.

        \begin{center}
          \footnotesize
          \begin{tabular}{c \tableColSmallSpace c \tableColSmallSpace c \tableColSmallSpace c \tableColSmallSpace c \tableColSmallSpace c}
            \toprule
	    \texttt{wellbore\_s}  & \texttt{year} & \texttt{month} & \texttt{day} &  \texttt{r\_existence\_kd\_nm} & \texttt{well\_s} \tableRowBreak
            \midrule
	    002         & 2010 &  04 & 01 & historic & 1 \tableRowBreak
	    002         & 2009 &  04 &  01 &     actual &   1 \tableRowBreak
            \bottomrule
          \end{tabular}
        \end{center}
        The mapping defining \owl{\small :completionDate}
        (c.f. Example~\ref{fig:mappingsEPDS}) uses the view
        \owl{\small wellbore} and has a filter \owl{\small
          r\_exis\-tence\-\_kd\_nm='actual'}. Observe that there is
        no FD (\owl{\small wellbore\_s}~$\rightarrow$~\owl{\small year
          month day}). However, the \vfd
        $\textowl{:Wellbore-\{\}} \mapsto^b \textowl{:completionDate}$
        is satisfied with this data instance, since in $\G^\O$ the
        wellbore \owl{\small :Wellbore-002} is connected to a single
        date \owl{\small "2010-04-01"\^{}\^{}xsd:date} through
        \textowl{:completionDate}.
\end{example}
Functional dependencies satisfied in a database instance often do not
correspond to any \vfd at the virtual level. We show this with an example:

\begin{example}\label{ex:opti3}
  Consider the table $T_1(x,y,z)$ with a single tuple: $(1,2,3)$.
  Clearly $x\rightarrow y$ and $x\rightarrow z$ are \fds satisfied in
  $T_1$.  Now consider the following mappings:
\begin{tabular}{l l}
\end{tabular}
\begin{lstlisting}[basicstyle=\ttfamily\scriptsize,mathescape,frame=none,frame=tb]
$\owlind{:\{x\}} \ \owl{P}_1 \ \owlind{:\{y\}}$ $\leftarrow \SQL{\SELECT{} * \FROM{} }$ $T_1~~~~~~~~~~\owlind{:\{x\}} \ \owl{P}_1 \ \owlind{:\{z\}}$ $\leftarrow \SQL{\SELECT{} * \FROM{} }$ $T_1$
\end{lstlisting}
Clearly, there is no \vfd involving $P_1$.
\end{example}

Hence, the shape of the mappings affects the
satisfiability of \vfds. Moreover, the ontology can also affect
satisfiability. We show this with an example:

\begin{example}
  Consider again the data instance $D_E$ from Example~\ref{ex:opti3},
  and the mappings~$\M_E$
  \begin{center}
  \begin{lstlisting}[basicstyle=\ttfamily\scriptsize,mathescape,frame=none,frame=tb]
$\owlind{:\{x\}} \ \owl{P}_1 \ \owlind{:\{y\}}$ $\leftarrow \SQL{\SELECT{} * \FROM{} }$ $T_1~~~~~~\owlind{:\{x\}} \ \owl{P}_2 \ \owlind{:\{z\}}$ $\leftarrow \SQL{\SELECT{} * \FROM{} }$ $T_1$
  \end{lstlisting}
  \end{center}
  Consider an OBDA instance $\O_E = ((\emptyset, \M_E, \Sigma_E) D_E)$. Then the virtual
  functional dependencies \owl{:\{\} }$ \mapsto^b P_1$ and \owl{:\{\}
  }$\mapsto^b P_2$ are satisfied in $\O$.  Consider another OBDA
  instance $\O_E' = ((\T_E, \M_E, \Sigma_E), D_E)$, where
  $\T_E = \{P_1 \owl{ rdfs:subClassOf } P_2\}$.  Then the two \vfds
  above are not satisfied in $\O_E'$.
\end{example}


\subsubsection{\vfd Based Optimization}
In this section we show how to optimize queries using \vfds.  Due to
space limitations, we focus on branching \vfds. The results for path
\vfds are analogous and can be found in the 
\ifextendedversion
appendix,
\else
technical report~\cite{techreport},
\fi
 as well as proofs.


\begin{definition}
  The set of mappings  
  $\M$ is \emph{basic for $\T$} if, for each property $P$ in $\T$, $P$
  is defined by at most one mapping in $\M_\T$. We say that $\O$ is basic if $\M$ is basic for $\T$.
\end{definition}

To ease the presentation, from now on we assume $\O$ to be basic.
We denote the (unique) mapping for $P_i$ in $\T$, $i \in \{1, \ldots, m\}$, 
as
\begin{equation*}\label{eq:Shapemapping}
    \begin{array}{l}
      t^i_d(\textbf{x}_i) ~~~P_i~~~t^i_r(\textbf{y}_i) \leftarrow sql_i(\textbf{z}_i).\\
    \end{array}
  \end{equation*}
  where $t^i_d$, and $t^i_r$ are templates for the domain and range of
  $P_i$, and $\textbf{x}_i$, $\textbf{y}_i$ are lists of attributes in
  $\textbf{z}_i$. The list $\textbf{z}_i$ is the list of projected
  attributes, which we assume to be the maximal list of attributes
  that can be projected from $sql_i$.

  Although we only consider basic instances, we show in the
\ifextendedversion
appendix
\else
technical report~\cite{techreport}
\fi
 how the results from this section
  can also be applied to the general case.

  \ignore{ For the sake of simplicity, in this section we assume fix
    restrict our OBDA instances to the case where each property $P_i$,
    $i \in \{1 \ldots m\}$, is defined by a single $\T$-mapping of the
    form:
    \begin{equation*}\label{eq:Shapemapping}
    \begin{array}{l}
      t^i_d(\textbf{x}_i) ~~~P_i~~~t^i_r(\textbf{y}_i) \leftarrow sql_i(\textbf{z}_i)\\
    \end{array}
\end{equation*}
In \cite{techreport} we show how to apply the results from this
section to the case where properties and classes are defined by more
than one $\T$-mapping.  }
We also assume that queries $sql_i(\vec{z}_i)$
always contain a filter expression of the form
$\sigma_{\text{notNull}(\textbf{x}_i,\textbf{y}_i)}$, even if we do not
specify it explicitly in the examples, since URIs cannot be
generated from nulls~\cite{W3Crec-R2RML}. Without loss of generality, we assume that
$\textbf{z}_1$ contains all the attributes in
$\textbf{x}_1$,$\textbf{y}_1, \ldots, \textbf{y}_n$.
  
In order to check satisfiability for a \vfd in an OBDA instance one can
analyze the DB based on the mappings and the ontology.
The next lemma formalizes this intuition.

\begin{lemma} \label{lm:detect1} Let $P_1,\dots, P_n$ be properties in $\T$ such
  that, for each $1 \le i < n$, $t_d^i=t_d^1$. Then, the \vfd
  $t_d^1\mapsto^b{P_1\dots P_n}$ is satisfied in $\O$ if and only if,
  for each $1 \le i \le n$, the \fd
  $\textbf{x}_{i}\rightarrow \textbf{y}_{i} $ is satisfied on
  $sql_i(\textbf{z}_i)^D$.
 \end{lemma}

 \begin{example}\label{ex:vfd1}
   Consider the properties \owl{:inWell} and \owl{:completionDate}
   from our running example.  The lemma above suggests that the \vfd
   ${\owl{:Wellbore-\{\}}}\mapsto^b{
     \owl{:isInWell}~\owl{:completionDate}}$ is satisfied in our OBDA
   instance with a database instance $D$ if and only if \emph{(i)}
   \owl{wellbore\_s}$\rightarrow$\owl{well\_s} is satisfied in
   $sql_{:isInWell}^D$, and \emph{(ii)}
   \owl{wellbore\_s}$\rightarrow$\owl{year month day} is satisfied in
   $sql_{:completionDate}^D$.

   From Example~\ref{ex:intro}, there is an organization constraint
   for the view \owl{wellbore} forcing only one completion date for
   each ``actual'' wellbore. As a consequence, the two FDs \emph{(i)}
   and \emph{(ii)} hold in any database $D$ following this
   organization constraint. Therefore, the \vfd in such instance is
   also satisfied.

 \end{example}

 
\medskip

We now show how  \vfds can be used to \emph{find redundant joins} that can be eliminated
in the  SQL translations. 

\ignore{
Given a branching \vfd of the form  ${t_d}\mapsto^b{P_1\dots P_n}$ 
that is satisfied in a virtual RDF graph where the data instance
satisfies certain inclusion assertions between the SQL queries in the mappings, 
we can safely use the SQL query defining $P_1$ 
in place of the joins between the SQL queries
defining $P_i$ ($i=1\dots n$).
The next lemma formalizes this intuition.
}

 \begin{definition}[Optimizing Branching \vfd]\label{def:opt1}
   Let $t$ be a template.
   An \emph{optimizing branching \vfd} is an expression of the form
   $t \rightsquigarrow^b P_1 \cdots P_n$.
   An optimizing \vfd \emph{$t \rightsquigarrow^b P_1 \cdots P_n$ is satisfied in $\O$} if $t \mapsto^b P_1 \cdots P_n$ is satisfied in $\O$, and for each $i \in \{1, \ldots, n\}$ it holds 
   \begin{equation}\label{eq:optimizingBranchingPrec}
     \pi_{ \textbf{x}_{1},\textbf{y}_{i}}\text{sql}_1(\textbf{z}_1)^D \subseteq 
     \rho_{\textbf{x}_{1}/\textbf{x}_{i}}(\pi_{ \textbf{x}_{i},\textbf{y}_{i}}\text{sql}_i(\textbf{z}_i))^D
   \end{equation}
\end{definition} 

 \begin{example}\label{ex:epds6}

   Recall that the \vfd
   ${\owl{:Wellbore-\{\}}}\mapsto^b{ \owl{:isInWell} ,
     \owl{:completionDate}}$ in Example~\ref{ex:vfd1} is
   satisfied in our OBDA instance.
%
  The precondition~(\ref{eq:optimizingBranchingPrec}) holds because
  (a) the properties are defined by the same SQL query (modulo
  projection) and (b) the organization constraint ``each wellbore
  entry must contain the information about name, date, and well (no
  nulls)''.
Thus, the optimizing \vfd ${\owl{:Wellbore-\{\}}}\rightsquigarrow^b{ \owl{:isInWell},\owl{:completionDate}}$ is satisfied in this instance.
\end{example}
 
 \begin{lemma}\label{def:opt1}
   Consider $n$ properties $P_1,\dots,P_n$ with $t_{d}^i = t_d^1$, for each $1 \le i \le n$, and for which 
   ${t_d^1}\rightsquigarrow^b{P_1\cdots P_n}$ is satisfied in $\O$. Then 
   {\footnotesize
     \begin{align*}
       \pi_{\gamma}(sql_1(\textbf{z}_1))^D = ~& 
\pi_{ \gamma} (sql_1(\textbf{z}_1) \Join_{\textbf{x}_{1}=\textbf{x}_{2}} sql_2(\textbf{z}_2) \Join \cdots \Join_{\textbf{x}_{1}=\textbf{x}_{n}} sql_n(\textbf{z}_n))^D, 
     \end{align*}
   }
where $\gamma=\textbf{x}_{1},\textbf{y}_{1}, \ldots, \textbf{y}_{n}$.

\end{lemma}

We now show how virtual functional dependencies can be used
in presence of triple patterns of the form $\textowl{?z} \textowl{ rdf:type } \textowl{C}$. 
As for properties, We assume that for each concept  $C_j$ we have a single $\T$-mapping of the form 
$C_j(t^j(\textbf{x})) \leftarrow sql_j(\textbf{z}_j)$.
\ignore{
\begin{definition}[Optimizing Classes]
We say that $C_j$ is optimizing w.r.t. the domain of $P_i$ if and only if
$t_j=t^i_d$ and\\
\[
\pi_{x}
sql_j(\textbf{z}_j)^D  
\supseteq 
\rho_{  \textbf{x}_{i}\mapsto x}
(\pi_{\textbf{x}_{i}}sql_i(\textbf{z}_i))^D
\]
We say that $C_j$ is optimizing w.r.t. the range of $P_i$ if and only if
$t_j=t^i_r$ and
\[
\pi_{x}
sql_j(\textbf{z}_j)^D 
\supseteq 
\rho_{  \textbf{y}_{i}\mapsto x}
(\pi_{\textbf{y}_{i}}sql_i(\textbf{z}_i))^D
\]
\end{definition}
}
\begin{definition}[Domain Optimizing Class Expression]
  A \emph{domain optimizing class expression} (domain \oce) is an
  expression of the form $t_j \rightsquigarrow^d_{P_i} C_j$. We say that
  $t_j \rightsquigarrow^d_{P_i} C_j$ is satisfied in $\O$ if
  $t_j=t^i_d$ and
  $\pi_{x}sql_j(\textbf{z}_j)^D \supseteq \rho_{x/\textbf{x}_{i}}
  (\pi_{\textbf{x}_{i}}sql_i(\textbf{z}_i))^D$.
\end{definition}

\begin{definition}[Range Optimizing Class Expression]
  A
  \emph{range optimizing class expression} (range \oce) is an
  expression of the form $t_j \rightsquigarrow^r_P C_j$. We say that
  $t_j \rightsquigarrow^r_{P_i} C_j$ is satisfied in $\O$ if
  $t_j=t^i_r$ and
  $\pi_{\textbf{x}}sql_j(\textbf{z}_j)^D \supseteq \rho_{\textbf{x}/\textbf{y}_{i}}
  (\pi_{\textbf{y}_{i}}sql_i(\textbf{z}_i))^D$ .
\end{definition}

Optimizing \vfds and classes give us a tool to identify those BGPs whose SQL translation can be optimized by removing redundant joins.

\begin{definition}[Optimizable branching BGP]\label{def:optimizable-branching-bgp}
   A BGP $\beta$ is optimizable w.r.t.  $\mathfrak{v} = {t_d}\rightsquigarrow^b{P_1\dots P_n}$ 
  if \emph{(i)} $\mathfrak{v}$ is satisfied in $\O$;
  \emph{(ii)} the BGP of triple patterns in $\beta$ involving properties is of the form 
$\textowl{?v P$_1$ ?v$_1$.~\dots ?v P$_n$ ?v$_n$.}$; 
and \emph{(iii)} for each triple pattern of the form $\textowl{?u rdf:type C}$ in $\beta$
, $\textowl{?u}$ is either the subject of some 
$P_i$ and
$t_d^i \rightsquigarrow^d_{P_i} C$ is satisfied in $\O$
, or  $?u$ is  in the object of some
 $P_i$ and $t_r^i \rightsquigarrow^r_{P_i} C$ is satisfied in $\O$.
\end{definition}

Finally, we prove that the standard SQL translation of optimizable BGPs contains redundant SQL joins that can be safely removed.

\begin{theorem}\label{th:finalFD}
  Let $\beta$ be an
  optimizable BGP w.r.t. ${t_d}\rightsquigarrow^x {P_1\dots P_n}$
  ($x=b,p$) in $\O$.  Let
  $\pi_{v/t_{d}^1,v_1/t_{r}^1,\ldots,v_n/t_{r}^n}sql_{\beta}$ be
  the SQL translation of $\beta$ as explained in
  Section~\ref{sec:sparql-qa}.  Let
  $sql_{\beta}'
  =sql_1(\textbf{x}_{1},\textbf{y}_{1}\dots,\textbf{y}_{n}) $.
Then  $sql_{\beta}^D$ and $sql_{\beta}'^D$  return the same answers.
\end{theorem}

\ignore{
  Observe that if we remove in Definition~\ref{def:opt1} the requirement about
  the  relational functional dependencies, then it is  
  not guaranteed that the \vfd can be used for optimization. 
  Consider the two  mappings for $P_1$, $P_2$ in Example~\ref{ex:opti3}.
  If $T_1$ now contains $(1,null,3)$ and $(1,2,null)$, although the query containment between 
  $sql_{P_1}$ and $sql_{P_2}$ holds (the queries are identical), 
  the self-join over $T_1$ 
  to obtain the (SPARQL) join of $P_1$ and $P_2$ triples patterns is not redundant\dl{I cannot find a place to this.}. 
}

\begin{corollary}
Let $Q$ be a SPARQL query.
Let  $sql_{Q}$ be the SQL translation of $Q$ as explained in Section~\ref{sec:sparql-qa}.
Let $sql'_{Q}$ be the SQL translation of $Q$ where all the SQL expressions corresponding to an optimizable BGPs w.r.t.
a set of \vfds have been optimized as stated
in Theorem~\ref{th:finalFD}.
Then  $sql_{Q}^D$ and $sql_{Q}^{'D}$  return the same answers.
\end{corollary}

\begin{example}
It is clear that the class 
\textowl{:Wellbore} is  optimizing w.r.t.  the domain of 
\textowl{:completion\-}\textowl{Date} and  \textowl{:is}\-\textowl{InWell}.
Since $\owl{:Wellbore-\{\}}\rightsquigarrow^b{\owl{:completionDate}, \owl{:isInWell}} $
is satisfied (c.f. Example~\ref{ex:epds6}),
one can allow the semantic optimizations to safely remove redundant joins
in query $sql_1$, sketched in Example~\ref{ex:intro}.
From Theorem~\ref{th:finalFD}, it follows that,
$sql_{:Wellbore} \Join sql_{:completionDate} \Join sql_{:isInWell}$
can be by simplified  to  $sql_{:Wellbore}$. 
\end{example}

\subsection{Enriching the OBDA Specification with Constraints}

We propose to enrich the traditional OBDA specification with a
constraint component, so as to allow the OBDA system to perform
enhanced optimization as described in the previous section.  More
formally, an \emph{OBDA specification with constraints} is a tuple
$\S_{constr} = (\S, \C)$ where $\S$ is an OBDA specification and $\C$
is a set of exact mappings, exact predicates, optimizing virtual
functional dependencies, and optimizing class expressions. An
\emph{instance of $\S_{constr}$} is an OBDA instance of $\S$
satisfying the constraints in $\C$.  Our intention is to be able to
use more of the constraints that exist in real databases for query
optimization, since we often see that these cannot be expressed by
existing database constraints (i.e. keys).  Since $\S$ does not
necessarily imply $\C$, checking the validity of $C$ may have to take
into account more information than just $S$.  The constraints $C$ may be
known to hold e.g. by policy, or be enforced by external tools, e.g.,
as in the case mentioned in the experiments below, by the tool used to
enter data into the database. 

In order to aid the user in the specification of $\C$, we implemented
tools to identify what exact mappings and optimizing virtual
functional dependencies are satisfied in a given OBDA instance
\ifextendedversion
(see appendix).
\else
(see~\cite{techreport}).
\fi. The user can then verify whether these
suggested constraints hold in general, for example because they derive
from storage policies or domain knowledge, and provide them as
parameters to the OBDA system. The user intervention is necessary,
because constraints derived from actual data can be an artifact of the
current situation of the database. 

\vspace{0.3mm}
\noindent\textbf{Optimizing \vfd Constraints.}
We have implemented a tool that automatically finds a restricted type of optimizing 
\vfds satisfied in a given OBDA instance and we have extended \ontop  to complement semantic optimization using these \vfds.
This implementation aims to mitigate the problem of redundant self-joins resulting 
from reifying relational tables. 
Although this is a simple case, it is extremely common in practice
 and, as we show in our experiments in Section~\ref{sec:experiments}, this class of \vfds is powerful enough to sensibly
improve the execution times in real world scenarios. 

\noindent\textbf{Exact Predicates Constraints.}
We implemented a tool to find exact predicates, and we extended 
\ontop to optimize $\T$-mappings with them. 
%
For each predicate $P$ in the ontology $\T$ of an OBDA instance $\O$, the tool constructs the query $q$ that returns all the individual/pairs in $P$.
Then it evaluates $q$ in the two OBDA instances $\O$ and $((\emptyset, \M, \Sigma), D)$. If the  answers for $q$ coincide in both instances,
then $P$ is exact.

\section{Experiments}\label{sec:experiments}
In this section we present a set of experiments evaluating the
techniques described above. 
In 
\ifextendedversion
the appendix
\else
 \cite{techreport}
\fi
 we ran additional
controlled experiments using an OBDA benchmark built on top of the
Wisconsin benchmark~\cite{wisconsin}, and obtain similar results to
the ones here.





\paragraph{\textbf{Statoil Scenario}}

In this section we briefly describe the Statoil use-case,
and the challenges it presents for OBDA. 
At Statoil, users access several databases on a daily basis, and one of the most important ones is the Exploration and Production Data Store (EPDS) database. EPDS
is a large legacy SQL (Oracle 10g) database comprising over 1500 tables (some of them with up to 10 million tuples) and 1600 views. The complexity of the SQL schema of EPDS is such that it is counter-productive and error-prone to manually write queries over the relational database. Thus, end-users either use only a set of tools with predefined SQL queries to access the database, or interact with IT experts so as to formulate the right query. The latter process can take weeks.
This situation triggered the introduction of OBDA in Statoil in the context of the
Optique project [13]. In order to test OBDA at Statoil, the users provided 60 queries (in natural language) that are relevant to their job, and that cannot be easily performed or formulated at the
moment. The Optique partners formulated these queries in SPARQL, and handcrafted an ontology, and a set of mappings connecting EPDS to the ontology. The ontology contains 90 classes, 37 object properties, and 31 data properties; and there are more than 140 mappings. The queries have between 0 to 2 complex filter expressions
(with several arithmetic and string operations), 0 to 5 nested optionals, modifiers such as ORDER BY and DISTINCT, and up to 32 joins. 

\ignore{
Geoscientists at \statoil, the main Norwegian oil company, access several databases on a daily basis, but in this work
we focus on  the \emph{Exploration and Production Data Store (EPDS)}.\dl{For what reason?}
EPDS is a large legacy SQL (Oracle 10g) database 
comprising over 1500 tables (some of them with up to 10 million tuples) and 1600 views. 
The geoscientists provided 60 queries (in natural language) that
are relevant to their job, and that cannot be easily performed at the moment\dl{Also, that cannot be easily formulated in SQL directly}.
From the vocabulary used in the users' queries we handcrafted an ontology,
and  a set of mappings connecting  the relevant parts of EPDS and the ontology. 
This preliminary ontology contains 90 classes, 37
object properties, and 31 data properties; 
and there are more than 140 mappings.
Using this vocabulary we also built 60 SPARQL queries
modeling the original query catalog in natural language. The obtained queries are moderately complex, as they have an average size of 13 BGPs, and one query contains $5$ \texttt{OPTIONAL} operators.
}

\paragraph{\textbf{Experiment Results.}}
%


The queries were executed sequentially on a HP ProLiant server with 24
Intel Xeon CPUs (X5650 @ 2.67 GHz), 283 GB of RAM.
Each query was evaluated three times and we took the average.
We ran the experiments with 4 exact concepts and 15 virtual functional
dependencies, found with our tools and validated by database experts. 
The 60 SPARQL queries have been executed over \ontop
with and without the optimizations for exact predicates and virtual functional
dependencies.
%
%
We consider that a query times out if the average execution time is greater than 20 minutes.

\begin{table*}[tb]
\centering
\scriptsize
\caption{Results from the tests over EPDS. }
\label{table:resultsEPDS}

%
\begin{tabular}{l|r|r|r|r} 
 & std. opt. & w/VFD & w/exact predicates & w/both \\ \hline
Number of queries timing-out
 & $17$ 
 & $10$ 
 & $11$ 
 & $4$ \ignore{ \text{( 2 without UDF)}$ }
\\ Number of fully answered queries 
 & $43$ 
 & $50$ 
 & $49$ 
 & $56$ 
\\ Avg. SQL query length (in characters)
 & $51521$ 
 & $28112$ 
 & $32364$ 
 & $8954$ 
\\ Average unfolding time 
 & $3.929$ s 
 & $3.917$ s 
 & $1.142$ s 
 & $0.026$ s 
\\ Average total query exec. time with timeouts
 & $376.540$ s 
 & $243.935$ s 
 & $267.863$ s 
 & $147.248$ s 
\\ Median total query exec. time with timeouts
 & $35.241$ s 
 & $11.135$ s 
 & $21.602$ s 
 & $14.936$ s 
\\ Average successful query exec. time (without timeouts)
 & $36.540$ s 
 & $43.935$ s 
 & $51.217$ s 
 & $67.248$ s 
\\ Median successful query exec. time (without timeouts)
 & $12.551$ s 
 & $8.277$ s 
 & $12.437$ s 
 & $12.955$ s 
\\ Average number of unions in generated SQL
& 6.3
& 3.4
& 5.1
& 2.2
\\ Average number of tables joined per union in generated SQL
& 21.0
& 18.2
& 20.0
& 14.2
\\ Average total number of tables in generated SQL
& 132.7
& 62.0
& 102.2
& 31.4
\end{tabular}

\end{table*}

\begin{figure*}[tb]
\begin{tabular}{c}
    \begin{tabular}{c}
      \begin{minipage}{\textwidth}
        \hspace{.001\textwidth}
	\begin{tikzpicture}
	  \pgfplotsset{every axis/.append style={
	      font=\footnotesize,
	      semithick,
	      tick style={thick}}}
          \begin{axis}[
              name=bottom axis,
	      ybar=0pt,
              bar width=0.1cm,
	      height=2cm,
	      width=0.9\textwidth,
	      ytick=\empty,
              xticklabels=\empty,
              yticklabels=\empty,
              xmin=0,
              xmax=28,
              grid=major,
	      legend columns=-1,
              ymin=0, 
	      ymax=10000,
              axis x line*=bottom
	    ]
    
\addplot[fill=blue!60] coordinates {
(1, 10000)
(2, 10000)
(3, 10000)
(4, 10000)
(5, 10000)
(6, 10000)
(7, 10000)
(8, 10000)
(9, 10000)
(10, 10000)
(11, 10000)
(12, 10000)
(13, 10000)
(14, 10000)
(15, 10000)
(16, 10000)
(17, 10000)
(18, 10000)
(19, 10000)
(20, 10000)
(21, 10000)
(22, 10000)
(23, 10000)
(24, 10000)
(25, 10000)
(26, 10000)
(27, 10000)
};
\addplot[draw=red, pattern color = red, pattern = north west lines] coordinates {
(1, 10000)
(2, 10000)
(3, 10000)
(4, 10000)
(5, 10000)
(6, 10000)
(7, 10000)
(8, 10000)
(9, 10000)
(10, 10000)
(11, 10000)
(12, 10000)
(13, 10000)
(14, 10000)
(15, 10000)
(16, 10000)
(17, 10000)
(18, 10000)
(19, 10000)
(20, 10000)
(21, 10000)
(22, 10000)
(23, 10000)
(24, 10000)
(25, 10000)
(26, 10000)
(27, 10000)
};
	  \end{axis}
	  \begin{axis}[
              at=(bottom axis.north),
              anchor=south, yshift=\pgfkeysvalueof{/tikz/axis break gap},
	      ybar=0pt,
              bar width=0.1cm,
              ylabel={\scriptsize \hspace{-15pt}Query execution time},
	      height=3cm,
	      width=0.9\textwidth,
      	      xticklabels=\empty,
              yticklabels={0.1 s,10 s,,20 m},
              ytick={100,10000,1000000,1200000},
              grid=major,
	      legend columns=-1,
              xmin = 0,
              xmax=28,
              ymin=100, 
	      ymax=1200000,
              ymode=log,
              xticklabel=\empty,
              axis x line*=top, 
              after end axis/.code={
                \draw (rel axis cs:0,0) +(-1mm,-1mm) -- +(1mm,1mm)
                    ++(0pt,-\pgfkeysvalueof{/tikz/axis break gap})
                    +(-1mm,-1mm) -- +(1mm,1mm)
                    (rel axis cs:1,0) +(-1mm,-1mm) -- +(1mm,1mm)
                    ++(0pt,-\pgfkeysvalueof{/tikz/axis break gap})
                    +(-1mm,-1mm) -- +(1mm,1mm);
            }
	    ]
	    
\addplot[fill=blue!60] coordinates {
(1, 1200000.0)
(2, 31045.266333333333)
(3, 1200000.0)
(4, 1200000.0)
(5, 451.915)
(6, 1200000.0)
(7, 14567.415666666668)
(8, 86767.832)
(9, 10429.749333333333)
(10, 1200000.0)
(11, 1200000.0)
(12, 1200000.0)
(13, 1200000.0)
(14, 47674.744666666666)
(15, 1200000.0)
(16, 273440.7256666667)
(17, 34515.80866666667)
(18, 1200000.0)
(19, 1200000.0)
(20, 1200000.0)
(21, 1200000.0)
(22, 35965.282666666666)
(23, 198613.8313333333)
(24, 37489.41133333333)
(25, 360937.3196666667)
(26, 1200000.0)
(27, 1200000.0)
};
\label{epds_untuned}
\addplot[draw=red, pattern color = red, pattern = north west lines] coordinates {
(1, 1200000.0)
(2, 6061.1923333333325)
(3, 1200000.0)
(4, 9883.007)
(5, 504.1513333333333)
(6, 4826.358333333333)
(7, 4967.532666666667)
(8, 12108.967666666666)
(9, 2158.47)
(10, 15743.838000000002)
(11, 1193568.6693333334)
(12, 16483.41233333333)
(13, 21841.840666666667)
(14, 3967.9269999999997)
(15, 72346.72266666667)
(16, 3857.595333333333)
(17, 10953.243666666667)
(18, 29955.592666666664)
(19, 449867.08466666663)
(20, 15452.311666666666)
(21, 23960.463666666667)
(22, 7615.3876666666665)
(23, 12905.610333333332)
(24, 7916.365000000001)
(25, 116688.18533333333)
(26, 890227.6223333334)
(27, 105109.99733333332)
};
\label{epds_tuned}
 \label{tuned-full}
            
	  \end{axis}
	\end{tikzpicture}
      \end{minipage}
    \end{tabular}
    \\	
    \begin{minipage}{\textwidth}
      \hspace{.05\textwidth}
      \footnotesize{
	\ref{epds_untuned} standard optimizations~~~ 
	\ref{epds_tuned} standard optimizations + VFD + exact predicates~~~
      }
    \end{minipage}
  \end{tabular}
\caption{Comparison of query execution time with standard optimizations.
   Log. scale}
\label{fig:queryTimesBars}
\end{figure*}
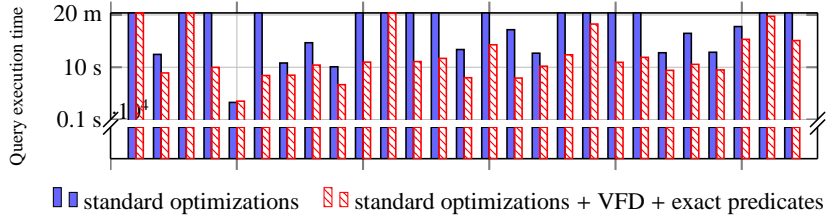


The results are summarized in Table \ref{table:resultsEPDS} and
Figure~\ref{fig:queryTimesBars}. We can see that 
the proposed optimizations allow \ontop to critically reduce the query
size and improve the performance of the query execution by orders of
magnitude. Specifically, in Figure~\ref{fig:queryTimesBars} we compare
standard optimizations with and without the techniques presented here.
Observe that the average successful query execution time is higher
with new optimizations than without because the number of successfully
executed queries increases.  
%
With standard optimizations, 17 SPARQL queries time out.
With both novel optimizations enabled, only four queries still time
out. 

A total of $27$ SPARQL queries get a more compact
SQL translation with new optimizations enabled. 
The largest proportional decrease in size of the SQL
query is $94$\%, from $171k$ chars, to $10k$.  The largest
absolute decrease in size of the SQL is $408k$ chars.
Note that the number of unions in the SQL may decrease also only with
VFD-based optimization. Since the VFD-based optimization removes joins, more
unions may become equivalent
and are therefore removed. 
The maximum measured decrease in execution time is on a query that times
out with standard optimizations, but uses $3.7$ seconds with new optimizations.

\section{Related work}\label{sec:related}


Dependencies have been intensively studied in the context of
traditional relational databases~\cite{BeVa81}.
Our work is related to the one in~\cite{Weddell};  in
particular their notion of path functional dependency is  close to
the notion of path \vfd presented here.
However, they do not consider neither ontologies, nor databases,
and their dependencies are not meant to be used to optimize queries.
There are a number of studies   on functional dependencies in RDF~\cite{yu2011extending,He-2014},
but as shown in Example~\ref{ex:opti3},
functional dependencies in RDF  do not necessarily correspond to a \vfd (when considering the ontology). 
Besides,  these works do not tackle the issue of SQL query optimization.


The notion of \emph{perfect
  mapping}~\cite{DLLM*13} is strongly related to the notion of exact
mapping. However there is a substantial difference: a perfect mapping
must be \emph{entailed} by the OBDA specification, whereas exact
mappings are additional constraints that enrich the OBDA
specification. For instance, perfect mappings would not be effective
in the \statoil use case, where organizational constraints and storage
policies are not entailed by the OBDA specification.  
The notion of \emph{EBox}~\cite{Rosa12b,optiqueprojectobda2014kyrie2} was proposed as an attempt to include 
constraints in OBDA. However, EBox axioms are defined through a  
$\T$-box like syntax. These axioms cannot express constraints
based on templates like virtual functional dependencies.

\section{Conclusions}
\label{sec:conclusions}

In this work we presented two novel optimization techniques for OBDA
that complement standard optimizations in the area, and enable
efficient SPARQL query answering over enterprise relational data.  We
provided theoretical foundations for these techniques based on two
novel OBDA constraints: virtual functional dependencies, and exact
predicates.
We implemented these techniques in our OBDA system \ontop and
empirically showed their effectiveness through extensive experiments that
display improvements on the query execution time up to orders of
magnitude.  


\vspace{2mm}
\noindent\textbf{Acknowledgement.} This work is partially supported by the EU
under IP project Optique (\emph{Scalable End-user Access to Big Data}), grant
agreement n.~FP7-318338.

\bibliographystyle{abbrv}
\bibliography{string-tiny,local-bib,krdb,w3c}

\ifextendedversion

\clearpage
\appendix

\section{Appendix}

\newcommand{\extV}[1][V]{\textit{ext}_{#1}}
\newcommand{\sm}{\lambda}

\newcommand{\bind}{\textsc{Bind}}

 \newcommand{\smerge}{\oplus}

\subsection{Background On SPARQL to SQL}

In this section, we recap the complete SPARQL to SQL translation~\cite{KRRXZ14}. This background will be used for the proofs in the following sections.



\subsubsection{SPARQL under Simple Entailment}
\label{sec:sparql-1}

SPARQL is a W3C standard language designed to query RDF graphs.
Its vocabulary contains four pairwise disjoint and countably infinite sets of symbols: \myIRI{} for \emph{IRIs},  \myBNK{} for \emph{blank nodes}, \myLIT{} for \emph{RDF literals}, and \myVAR{} for \emph{variables}. The elements of $\sDOM = \myIRI \cup \myBNK \cup \myLIT$ are called \emph{RDF terms}.
A \emph{triple pattern} is an element of $(\sDOM\cup\myVAR)\times (\myIRI\cup\myVAR) \times (\sDOM\cup\myVAR)$. A \emph{basic graph pattern} (\emph{BGP}) is a finite set of  triple patterns. Finally, a \emph{graph pattern}, $P$, is an expression defined by the grammar 
\begin{align*}
P \ &::=  \ \textsc{BGP} \ \mid \  \filter(P,F)  \ \mid \ \bind(P,v,c) \ \mid  \ \union(P_1,P_2) \\ 
 & \qquad \mid \ \join(P_1,P_2) \ \mid \ \leftjoin(P_1,P_2,F),
\end{align*}
where $F$, a \emph{filter}, is a formula constructed from atoms of the form $\textit{bound}(v)$,  \mbox{$(v=c)$}, $(v=v')$, for $v,v' \in \myVAR$, $c \in \sDOM$, and possibly other built-in predicates using the logical connectives $\land$ and $\neg$. The set of variables in $P$ is denoted by $\var(P)$.

A \emph{SPARQL query} is a graph pattern $P$  with a \emph{solution modifier}, which specifies the \emph{answer variables}---the variables in $P$ whose values we are interested in---and the form of the output (we ignore other solution modifiers for simplicity). 
The values to variables are given by \emph{solution mappings}, which are \emph{partial} maps $s  \colon \myVAR \to \sDOM$ with (possibly empty) domain $\dom (s)$. 
In this paper, we use the set-based (rather than bag-based, as in the specification) semantics for SPARQL.
For sets $S_1$ and $S_2$ of solution mappings, a filter $F$, a variable $v\in \myVAR$ and a term $c\in \sDOM$, let
\begin{itemize}
\item $\filter(S,F) = \{ s \in S \mid F^s = \top \}$;
\item $\bind(S,v,c) = \{ s \smerge \{ v \mapsto c \} \mid s\in S \}$ (provided that $v\notin\dom(s)$, for $s\in S$);
\item $\union(S_1, S_2) = \{ s \mid s\in S_1 \text{ or } s\in S_2 \}$;
\item $\join(S_1, S_2) = \{ s_1 \smerge s_2 \mid s_1 \in S_1 \text{ and } s_2 \in S_2 \text{ are compatible}\}$;
\item $\leftjoin(S_1,S_2,F) =  \filter(\join(S_1,S_2), F) \ \cup \ \{ s_1 \in S_1 \mid \text{ for all } s_2\in S_2,$\\
\mbox{}\hfill $\text{ either } s_1, s_2 \text{ are incompatible or } 
F^{s_1\smerge s_2} \ne \top \}$. 
\end{itemize}
Here, $s_1$ and $s_2$ are \emph{compatible} if $s_1(v) = s_2(v)$, for any $v \in \dom(s_1) \cap \dom(s_2)$, in which case $s_1\smerge s_2$ is a solution mapping with $s_1 \smerge s_2\colon v \mapsto s_1(v)$, for $v \in \dom(s_1)$, $s_1 \smerge s_2\colon v \mapsto s_2(v)$, for $v \in \dom(s_2)$, and domain $\dom(s_1) \cup \dom(s_2)$. The \emph{truth-value $F^s \in \{\top, \bot, \varepsilon\}$ of a filter $F$ under a solution mapping $s$} is defined inductively: 
\begin{itemize}
\item $(\textit{bound}(v))^s$ is $\top$ if $v \in \dom(s)$ and $\bot$ otherwise; 

\item $(v = c)^s = \varepsilon$ if $v\notin\dom(s)$; otherwise, 
$(v = c)^s$ is the classical truth-value of the predicate $s(v) = c$;  similarly, $(v = v')^s = \varepsilon$ if either $v$ or $v'\notin\dom(s)$; otherwise, 
$(v = v')^s$ is the classical truth-value of the predicate $s(v) = s(v')$;

\item $(\neg F)^s=\begin{cases}
\varepsilon, & \text{if } F^s = \varepsilon,\\[-2pt]
\neg  F^s, & \text{otherwise,}
\end{cases}$ 
\quad and
$(F_1 \land F_2)^s = \begin{cases}%
\bot, & \text{if } F_1^s =\bot \text{ or } F_2^s =\bot,\\[-2pt]
\top, & \text{if  } F_1^s =F_2^s =\top,\\[-2pt]
\varepsilon, &  \text{otherwise.}
\end{cases}$
\end{itemize}
Finally, given an RDF graph $G$, the \emph{answer to a graph pattern $P$ over $G$} is the set $\sANS{P}{G}$ of solution mappings defined by induction using the operations above and starting from the following base case: for a basic graph pattern $B$, 
\begin{equation}\label{ans1}
\sANS{ B }{G}
 = \{ s\colon \var(B) \to \sDOM \mid s(B) \subseteq G \},
\end{equation}
where $s(B)$ is the set of triples resulting from substituting each variable $u$ in $B$ by $s(u)$. This semantics is known as \emph{simple entailment}.

\subsubsection{Translating SPARQL under Simple Entailment to SQL}\label{sec:SQL}


We recap the basics of relational algebra and SQL (see e.g.,~\cite{AbHV95}). 
Let $U$ be a finite (possibly empty) set of \emph{attributes}. A \emph{tuple over} $U$ is a map $t \colon U \to \Delta$, where $\Delta$ is the underlying domain, which always contains a distinguished element $\Null$. 
A ($|U|$-\emph{ary}) \emph{relation over $U$} is a finite set of tuples over $U$ (again, we use the set-based rather than bag-based semantics). 
A \emph{filter $F$ over $U$} is a formula constructed from atoms $\isNull(U')$, $(u=c)$ and $(u=u')$, where $U'\subseteq U$, $u,u' \in U$ and $c \in \Delta$, using the connectives $\land$ and
$\neg$. Let $F$ be a filter with variables $U$ and let $t$ be a tuple over $U$. The \emph{truth-value $F^t \in \{\top, \bot, \varepsilon\}$ of $F$ over $t$} is defined inductively: 
\begin{itemize}
\item $(\isNull(U'))^t$ is $\top$ if $t(u)$ is $\Null$, for all $u\in U'$, and $\bot$ otherwise; 

\item $(u = c)^t = \varepsilon$ if $t(u)$ is $\Null$; otherwise, 
$(u = c)^t$ is the classical truth-value of the predicate $t(u) = c$; similarly, $(u = u')^t = \varepsilon$ if either $t(u)$ or $t(u')$ is $\Null$; otherwise, 
$(u = u')^t$ is the classical truth-value of the predicate $t(u) = t(u')$;

\item $(\neg F)^t=\begin{cases}
\varepsilon, & \text{if } F^t = \varepsilon,\\[-1pt]
\neg  F^t, & \text{otherwise,}
\end{cases}$ 
\quad and 
$(F_1 \land F_2)^t = \begin{cases}%
\bot, & \text{if } F_1^t =\bot \text{ or } F_2^t =\bot,\\[-2pt]
\top, & \text{if  } F_1^t =F_2^t =\top,\\[-2pt]
\varepsilon, &  \text{otherwise.}
\end{cases}$
\end{itemize}
(Note that $\neg$ and $\land$ are interpreted in the same three-valued logic as in SPARQL.)
We use standard relational algebra operations such as union, difference, projection, selection, renaming and natural (inner) join. Let $R_i$ be a relation over $U_i$, $i=1,2$. 
\begin{itemize}
\item If $U_1 = U_2$ then the standard 
$R_1 \cup R_2$ 
and $R_1 \setminus R_2$ 
are relations over $U_1$.

\item If $U \subseteq U_1$ then $\pi_{U}R_1 = R_1 |_U$ is a relation over $U$.

\item If $F$ is a filter over $U_1$ then $\sigma_F R_1 = \{t\in R_1 \mid F^t = \top\}$ is a relation over $U_1$.

\item If $v \notin U_1$ and $u \in U_1$ then 
$\rho_{v/u}R_1 = \bigl\{t_{v/u} \mid  t\in R_1\bigr\}$, where $t_{v/u}\colon v \mapsto t(u)$ and $t_{v/u}\colon u' \mapsto t(u')$, for $u' \in U_1\setminus \{u\}$, is a relation over $(U_1\setminus \{u\}) \cup \{v\}$.

\item $R_1 \Join R_2 = \{t_1 \smerge t_2  \mid  t_1 \in R_1 \text{ and } t_2 \in R_2 \text{ are compatible} \}$  is a relation over $U_1 \cup U_2$. Here, $t_1$ and $t_2$ are \emph{compatible} if $t_1(u) = t_2(u) \ne \Null$, 
for all $u \in U_1 \cap U_2$, in which case  
a tuple $t_1 \smerge t_2$ over $U_1 \cup U_2$  is defined by taking $t_1 \smerge t_2\colon u\mapsto t_1(u)$, for $u \in U_1$, and $t_1 \smerge t_2\colon u \mapsto t_2(u)$, for $u \in U_2$ (note that if $u$ is $\Null$ in either of the tuples then they are incompatible).
\end{itemize}
To bridge the gap between partial functions (solution mappings) in SPARQL and total mappings (on attributes) in SQL, we require one more operation (expressible in SQL):
\begin{itemize}
\item If $U \cap U_1 = \emptyset$ then the \emph{padding} $\mu_{U} R_1$ is $R_1 \Join \Null^{U}$, where $\Null^U$  is the relation consisting of a single tuple $t$ over $U$ with $t\colon u \mapsto \Null$, for all $u \in U$.
\end{itemize}
By an \emph{SQL query}, $Q$, we understand any expression constructed from relation symbols (each over a fixed set of attributes) and filters using the  relational algebra operations given above (and complying with all restrictions on the structure). 
Suppose $Q$ is an SQL query and $D$ a data instance which, for any relation symbol in the schema under consideration, gives a concrete relation over the corresponding set of attributes. The \emph{answer  to $Q$ over $D$} is a relation $\|Q\|_D$ defined inductively in the obvious way starting from the base case: for a relation symbol $Q$,  $\|Q\|_D$ is the corresponding relation in $D$.


We now define a translation, $\tra$, which, given a graph pattern $P$, returns an SQL query $\tra(P)$  with the same answers as $P$. More formally, for a set of variables $V$, let $\extV$ be a function transforming any solution mapping $s$ with $\dom(s)\subseteq V$ to a tuple over $V$ by padding it with $\Null$s:  
\begin{equation*}
\extV(s) \  = \ \{ v \mapsto s(v) \mid v\in \dom(s) \} \ \cup \ \{ v\mapsto \Null \mid v\in V \setminus\dom(s)\}.
\end{equation*}
The \emph{relational answer to $P$ over $G$} is $\| P\|_G = \{\extV[\var(P)](s) \mid s\in \sANS{P}{G}\}$. 
The SQL query $\tra(P)$ will be such that, 
for any RDF graph $G$, the relational answer to $P$ over $G$ coincides with the answer to $\tra(P)$ over $\textit{triple}(G)$, the database instance storing $G$ as a ternary relation $\textit{triple}$ with the attributes $\textit{subj}$, $\textit{pred}$, $\textit{obj}$. 
First, we define the translation of a SPARQL filter  $F$ by taking $\tra(F)$ to be the SQL filter  obtained by replacing each $\textit{bound}(v)$ with $\neg \isNull(v)$ (other built-in predicates can be handled similarly). 
\begin{proposition}\label{prop:filter}
Let $F$ be a SPARQL filter and let $V$ be the set of variables in $F$.
Then $F^s = (\tra(F))^{\extV[V](s)}$, for any solution mapping $s$ with $\dom(s)\subseteq V$.  
\end{proposition}

The definition of $\tra$ proceeds by induction on the construction of $P$. Note that we can always assume that graph patterns \emph{under simple entailment} do not contain blank nodes because they can be replaced by fresh variables. It follows that a BGP $\{ \textit{tp}_1,\dots,\textit{tp}_n\}$ is equivalent to $\join(\{\textit{tp}_1\},\join(\{\textit{tp}_2\},\dots))$. 
So, for the basis of induction we set
\begin{equation*}
\tra(\{\langle s,p,o\rangle\})  =  \begin{cases}%
\pi_\emptyset \sigma_{(\textit{subj} = s) \land (\textit{pred} = p) \land (\textit{obj} = o)}\,\textit{triple}, & \text{if } s,p,o\in \myIRI \cup \myLIT,\\ 
\proj{s} \rho_{s/\textit{subj}}\, \sigma_{(\textit{pred} = p) \land (\textit{obj} = o)}\,\textit{triple}, & \text{if } s\in\myVAR \text{ and } p,o\in \myIRI \cup \myLIT,\\ 
\proj{s,o} \rho_{s/\textit{subj}}\,\rho_{o/\textit{obj}} \,\sigma_{\textit{pred} = p}\,\textit{triple}, & \text{if } s,o\in\myVAR, s \ne o, p\in \myIRI \cup \myLIT,\\ 
\proj{s} \rho_{s/\textit{subj}} \,\sigma_{(\textit{pred} = p)\land (\textit{subj} = \textit{obj})}\,\textit{triple},\hspace*{-0.5em} & \text{if } s,o\in\myVAR, s = o,  p\in \myIRI \cup \myLIT,\\[-2pt]
\dots
\end{cases}
\end{equation*}
(the remaining cases are similar).  
Now, if $P_1$ and $P_2$ are graph patterns and $F_1$ and $F$ are filters containing only variables in $\var(P_1)$ and $\var(P_1) \cup \var(P_2)$, respectively, then we set  $U_i = \var(P_i)$, $i=1,2$, and
\begin{align*}
& \tra(\filter(P_1,F_1))  = \sigma_{\tra(F_1)} \tra(P_1),\\
& \tra(\bind(P_1,v,c)) =  \tra(P_1) \Join \{ v\mapsto c \},\\
& \tra(\union(P_1,P_2))  = \ \mu_{U_2 \setminus U_1} \tra(P_1) \ \cup \ \mu_{U_1 \setminus U_2}\tra(P_2),\\
& \tra(\join(P_1, P_2))  =  \hspace*{-1.5em}\bigcup_{\begin{subarray}{c}V_1, V_2 \subseteq U_1 \cap U_2\\V_1 \cap V_2 = \emptyset\end{subarray}}  \hspace*{-1.5em} \bigl[ (\pi_{U_1\setminus V_1} \sigma_{\textit{isNull}(V_1)} \tra(P_1)) \!\Join\! 
(\pi_{U_2\setminus V_2} \sigma_{\textit{isNull}(V_2)} \tra(P_2)) \bigr],\\
& \tra(\leftjoin(P_1,P_2,F)) = \ \  \tra(\filter(\join(P_1,P_2),F)) \ \cup {} \\
& \hspace*{5em} \mu_{U_2\setminus U_1} \bigl(\tra(P_1) \setminus \hspace*{-0.7em}  \bigcup_{V_1\subseteq U_1 \cap U_2} \hspace*{-1em}\mu_{V_1}\pi_{U_1\setminus V_1} \tra(\filter(\join(P_1^{V_1,U_1\cap U_2}, P_2), F))\bigr),
\end{align*}
where $P^{V,U} = \filter(P, \bigwedge_{v\in V} \neg\textit{bound}(v) \land \bigwedge_{v\in U\setminus V} \textit{bound}(v))$. It is readily seen that any $\tra(P)$ is a valid SQL query and defines a relation over $\var(P)$; in particular, $\tra(\join(P_1,P_2))$ is a relation over \mbox{$\bigcup_{i = 1,2}(U_i \setminus V_i)= U_1 \cup U_2 = \var(\join(P_1,P_2))$}.

\begin{theorem}\label{thm:sparql-sql}
For any RDF graph $G$ and any graph pattern $P$,  
$\| P\|_G = \|\tra(P)\|_{\textit{triple}(G)}.$
\end{theorem}

\subsubsection{R2RML Mappings}\label{sec:r2rml}

The SQL translation of a SPARQL query constructed above has to be evaluated over the ternary relation $\textit{triple}(G)$ representing the virtual RDF graph $G$. Our aim now is to transform it to an SQL query over the actual database, which is related to $G$ by means of an R2RML mapping~\cite{W3Crec-R2RML}.
We begin with a simple example.
\begin{example}\label{ex:map}
The following R2RML mapping (in the Turtle syntax) populates 
an object property \nm{ub:UGDegreeFrom} from a relational table \nm{students}, whose attributes \nm{id} and \nm{degreeuniid} identify graduate students and their  universities:\\[2pt]
{\small\sffamily
\hspace*{1em}\_:m1 a rr:TripleMap;\\
\hspace*{4em}rr:logicalTable [ rr:sqlQuery "SELECT * FROM students WHERE stype=1" ];\\ 
\hspace*{4em}rr:subjectMap [ rr:template "/GradStudent\{id\}" ] ; \\
\hspace*{4em}rr:predicateObjectMap [ rr:predicate ub:UGDegreeFrom ;\\
\hspace*{14.8em}rr:objectMap [ rr:template "/Uni\{degreeuniid\}" ] ] 	
}\\[2pt]
More specifically, for each tuple in the query, an R2RML processor generates an RDF triple with the predicate \nm{ub:UGDegreeFrom} and the subject and object  constructed from attributes \nm{id} and \nm{degreeuniid}, respectively, using IRI templates.
\end{example}

Our aim now is as follows: given an R2RML mapping $\M$, we are going to define an SQL query  $\smash{\rTwoTr_\M(\textit{triple})}$  that constructs the relational representation $\smash{\textit{triple}(G_{D,\M})}$ of the virtual RDF graph $\smash{G_{D,\M}}$ obtained by $\M$ from any given data instance $D$.
Without loss of generality and to simplify presentation, we assume that each triple map has\\[4pt] 
\hspace*{1em}-- one logical table (\nm{rr:sqlQuery}), \\
\hspace*{1em}-- one subject map (\nm{rr:subjectMap}), which does not have resource typing (\nm{rr:class}), \\ 
\hspace*{1em}-- and one predicate-object map with one \nm{rr:predicateMap} and one \nm{rr:objectMap}.\\[4pt]
This normal form can be achieved by introducing predicate-object maps with \nm{rdf:type} and splitting any triple map into a number of triple maps with the same logical table and subject. We also assume that triple maps contain no referencing object maps  (\nm{rr:parentTriplesMap}, etc.) since they can be eliminated using  joint SQL queries~\cite{W3Crec-R2RML}. 
Finally, we assume that the term maps (i.e., subject, predicate and object maps) contain no constant shortcuts and are of the form $[\nm{rr:column}\ v]$, $[\nm{rr:constant}\ c]$ or  $[\nm{rr:template}\ s]$.

Given a triple map $m$ with a logical table (SQL query) $R$, we construct a selection $\sigma_{\neg \isNull(v_1)} \cdots \sigma_{\neg \isNull(v_k)} R$, where $v_1,\dots,v_k$ are the  \emph{referenced columns} of $m$ (attributes of $R$ in the term maps in $m$)---this is done to exclude tuples that contain $\Null$~\cite{W3Crec-R2RML}. To construct $\rTwoTr_m$, the selection filter is prefixed with projection $\proj{\textit{subj},\textit{pred},\textit{obj}}$ and, for each of the three term maps, either with renaming (e.g., with $\rho_{\textit{obj}/v}$ if the object map is of the form  $[\nm{rr:column}\ v]$) or with value creation (if the term map is of the form $[\nm{rr:constant}\ c]$ or $[\nm{rr:template}\ s]$; in the latter case, we use the built-in string concatenation function \nm{\cc}).
For instance, the mapping \nm{\_:m1} from Example~\ref{ex:map} is converted to the SQL query\\[2pt]
{\small\sffamily
\hspace*{2em}SELECT ('/GradStudent' \cc{} id) AS subj, 'ub:UGDegreeFrom' AS pred,\\ 
\hspace*{2em}\phantom{SELECT }  ('/Uni' \cc{} degreeuniid) AS obj FROM students\\ 
\hspace*{2em}WHERE (id IS NOT NULL) AND (degreeuniid IS NOT NULL) AND (stype=1).
}\\[2pt]
Given an R2RML mapping $\M$, we set $\rTwoTr_{\M}(\textit{triple}) = \bigcup_{m\in \M}
\rTwoTr_m$.

\begin{proposition}\label{prop:r2rml}
For any R2RML mapping $\M$ and data instance $D$,  
$t\in \|\rTwoTr_{\M}(\textit{triple})\|_D$ if and only if $t\in \textit{triple}(G_{D,\M})$.
\end{proposition}

Finally, given a graph pattern $P$ and an R2RML mapping $\M$, we define $\rTwoTr_{\M}(\tra(P))$ to be the result of replacing every occurrence of the relation $\textit{triple}$ in the query $\tra(P)$, constructed in Section~\ref{sec:SQL}, with $\rTwoTr_{\M}(\textit{triple})$. By Theorem~\ref{thm:sparql-sql} and Proposition~\ref{prop:r2rml}, we obtain:

\begin{theorem}\label{thm:sparql-r2rml-sql}
For any graph pattern $P$, R2RML  mapping $\M$ and data instance $D$,
$\| P\|_{G_{D,\M}} = \|\rTwoTr_{\M}(\tra(P))\|_D$.
\end{theorem}



\subsection{Proofs of Section~\ref{sec:exact-mappings-obda}}

\begin{propositionnum} {\ref{prop:exact}}
  Let $\M'$ be exact for the predicate $A$ in $\T$.  
  Let $\M_\T'$ be the result of replacing all the mappings 
  defining $A$ in $\M_\T$ by $\M'$. 
	Then $\G^\O=\G^{(( \emptyset,\M_\T', \Sigma ), D)}$.
\end{propositionnum}

\begin{proof}[Sketch]
  By the definition of $\T$-mappings, we have
  $\G^{\O}=\G^{\emptyset,\M_\T,D}.$ For all predicates other than $A$,
  $\M_\T$ and $\M_\T'$ produce the same set of triples since the
  mappings defining them are identical. For the predicate $A$, since
  $\M'$ is exact in $\O$, $\M_\T$ and $\M_\T'$ also produce same set
  of triples. Therefore
  $\G^{\emptyset,\M_\T,D}=\G^{\emptyset,\M_\T',D}.$
\end{proof}

\subsection{Proofs of Section~\ref{s:fds-in-obda}}
\label{sec:proofs-sect-refs:fds}


\begin{lemmanum}{\ref{lm:detect1}} Let $P_1,\dots, P_n$ be properties in $\T$ such
  that, for each $1 \le i < n$, $t_d^i=t_d^1$. Then, the \vfd
  $t_d^1\mapsto^b{P_1\dots P_n}$ is satisfied in $\O$ if and only if,
  for each $1 \le i \le n$, the \fd
  $\vec{x}_{i}\rightarrow \vec{y}_{i} $ is satisfied on
  $sql_i(\vec{z}_i)^D$.
 \end{lemmanum}

 \begin{proof}
      \footnotesize
      ~\\
      $t_d^1\mapsto^b{P_1\dots P_n}$ is satisfied in $G^{\O}$\\
      $\Updownarrow$ (Definition~\ref{def:nfd})\\
      $\forall s \in S_{t_d^1}: \forall 1 \le i \le n : (s,o) \in P_i^{G^{\O}} \land (s,o') \in P_i^{G^{\O}} \Rightarrow o = o'$\\
      $\Updownarrow$ (Mappings assumptions for $P_i$)\\
      $\forall 1 \le i \le n : \forall \vec{u} \in \pi_{\vec{x}_i} \sigma_{notNull(\vec{x}_i,\vec{y}_i)}sql_i(\vec{z}_i)^D : (\vec{u},\vec{y}) \in \pi_{\vec{x}_i\vec{y}_i}\sigma_{notNull(\vec{x}_i,\vec{y}_i)}sql_i(\vec{z}_i)^D \land (\vec{u},\vec{y'}) \in \pi_{\vec{x}_i\vec{y}_i}\sigma_{notNull(\vec{x}_i,\vec{y}_i)}sql_i(\vec{z}_i)^D \Rightarrow \vec{y} = \vec{y'}$\\
      $\Updownarrow$ (Definition of Functional Dependency)\\
      $\forall 1 \le i \le n : \vec{x}_i \rightarrow \vec{y}_i \text{ is satisfied in } \pi_{\vec{x}_i\vec{y}_i}\sigma_{notNull(\vec{x}_i,\vec{y}_i)}sql_i(\vec{z}_i)^D$\\
      $\Updownarrow$ \\ 
      $\forall 1 \le i \le n : \vec{x}_i \rightarrow \vec{y}_i \text{ is satisfied in } \sigma_{notNull(\vec{x}_i,\vec{y}_i)}sql_i(\vec{z}_i)^D$
    \end{proof}
    $\hfill\Box$


\begin{lemmanum}{\ref{def:opt1}}
   Consider $n$ properties $P_1,\dots,P_n$ in $\T$ with $t_{d}^i = t_d^1$, for each $1 \le i \le n$, and for which 
   ${t_d^1}\rightsquigarrow^b{P_1\cdots P_n}$ is satisfied in $\O$. Then 
   {\footnotesize
     \begin{align*}
       \pi_{\gamma}(sql_1(\vec{z}_1))^D = ~& 
\pi_{ \gamma} (sql_1(\vec{z}_1) \Join_{\vec{x}_{1}=\vec{x}_{2}} sql_2(\vec{z}_2) \Join \cdots \Join_{\vec{x}_{1}=\vec{x}_{n}} sql_n(\vec{z}_n))^D, 
     \end{align*}
   }
where $\gamma=\vec{x}_{1},\vec{y}_{1}, \ldots, \vec{y}_{n}$.

\end{lemmanum}

    \begin{proof}
      The direction $\subseteq$ of the equality can be obtained easily. Here we prove the direction $\supseteq$.

      Let $q_{branch}^D$ denote the right hand side expression in the equality.
      Assume the containment $\supseteq$ does not hold. Then, this means there exists a tuple $(\vec{s}, \vec{v}_1, \ldots, \vec{v}_n)$ such that
      
      \begin{itemize}
      \item $(\vec{u}, \vec{v}_1, \ldots, \vec{v}_n) \in q_{branch}^D$,  and
      \item $(\vec{u}, \vec{v}_1, \ldots, \vec{v}_n) \notin  \pi_{ \vec{x}_1,\vec{v}_1, \ldots, \vec{v}_n}\sigma_{notNull(\vec{x}_1, \vec{y}_1)}sql_1(\vec{z}_1)^D$
      \end{itemize}

      The above implies that there exists an index $j$, $1 \le j \le n$, such that
      \begin{itemize}
      \item $(\vec{u}, \vec{v}_j) \in \pi_{\vec{x}_1,\vec{y}_j}q_{branch}^D$, and 
      \item $(\vec{u}, \vec{v}_j) \notin  \pi_{ \vec{x}_1,\vec{y}_j}\sigma_{notNull(\vec{x}_1, \vec{y}_1)}sql_1(\vec{z}_1)^D$
      \end{itemize}
      
      Then, we can distinguish three cases:

      \begin{enumerate}
      \item $\vec{u} \notin \pi_{\vec{x}_1}\sigma_{notNull(\vec{x}_1, \vec{y}_1)}sql_1(\vec{z}_1)$.
        
        Then $u \notin \pi_{\vec{x}_1}q_{branch}^D$, hence $(\vec{u}, \vec{v}_j) \notin \pi_{\vec{x}_1,\vec{y}_j}q_{branch}^D$; contradiction.
      \item $(\vec{u},\vec{v}_j') \in \pi_{\vec{x}_1,\vec{y}_j}\sigma_{notNull(\vec{x}_1, \vec{y}_1)}sql_1(\vec{z}_1)$, and $\vec{v}_j' = \texttt{null}$.
        
        Since ${t_d^1}\rightsquigarrow^b{P_1\cdots P_n}$ is satisfied in $\O$, it must be 
        $(\vec{u},\texttt{null}) \in \pi_{\vec{x}_j,\vec{y}_j}\sigma_{notNull(\vec{x}_j, \vec{y}_j)}sql_j(\vec{z}_j)^D$, which is impossible.
        
      \item $(\vec{u},\vec{v}_j') \in \pi_{\vec{x}_1,\vec{y}_j}\sigma_{notNull(\vec{x}_1, \vec{y}_1)}sql_1(\vec{z}_1)^D$, $\vec{v}_j \neq \vec{v}_j'$ and not $\vec{v}_j$ nor $\vec{v}_j'$ is \texttt{null}.

        This violates the hypothesis that ${t_d^1}\rightsquigarrow^b{P_1\cdots P_n}$ is satisfied in $\O$, because of Lemma~\ref{lm:detect1}.  
      \end{enumerate}
      Hence, by contradiction we conclude that the containment $\supseteq$ must hold.
    \end{proof}
    $\hfill\Box$

   \paragraph{\textbf{Results and Proofs for PATH VFDs}}

   ~\\

   \begin{lemma}\label{lm:detect2}
     Let $P_1,\dots, P_n$ be properties in $\T$ such that, for each $1 \le i \le n$, $t_{r_i} = t_{d_{i+1}}$.
     Then, the \vfd ${t^1_{d}}\mapsto^p{P_1\dots P_n}$ is satisfied in $\O$ if and only if 
    the \fd 
    $ \textbf{x}_{1}\rightarrow \textbf{y}_{1} \cdots \textbf{y}_{n}$ 
    is satisfied in: 
    \[\pi_{\textbf{x}_{1}\textbf{y}_{1}\cdots \textbf{y}_{n}}(sql_1(\vec{z}_1)) \Join_{\textbf{y}_{1}=\textbf{x}_{2}} sql_2(\vec{z}_2) \Join_{\textbf{y}_{2}=\textbf{x}_{3}}\dots  \Join_{\textbf{y}_{n-1}=\textbf{x}_{n}} sql_n(\vec{z}_n))^D\]

    \begin{proof}
      ~\\
      \scriptsize
      ${t_d^1}\mapsto^p{P_1\dots P_n}$ is satisfied in $G^{\O}$\\
      $\Updownarrow$ (Definition~\ref{def:nfd})\\
      $ \forall s\in S^1_{t_d} : \exists \text{ unique list } \langle o_1,\dots, o_n\rangle$ in $G^{\O}$ such that
      $ \{ (s,P_1,o_1), \dots, (o_{n-1}, P_n, o_n) \}$ $ \subseteq G^{\O}$ \\
      $\Updownarrow$ \\
      $ \forall s\in S^1_{t_d} : \forall o_1, \ldots, o_n, o_1', \ldots, o_n' \text{ in } G^{\O} : (s,o_1) \in P_1^{G^{\O}} \land \ldots \land (o_{n-1}, o_n) \in P_n^{G^{\O}} \bigwedge (s,o_1') \in P_1^{G^{\O}} \land \ldots \land (o_{n-1}', o_n') \in P_n^{G^{\O}} \Rightarrow o_1 = o_1', \ldots, o_n = o_n'$\\
      $\Updownarrow$ (Mappings assumptions for $P_i$)
      \begin{align*}
        \forall \vec{u} \in \pi_{\vec{x}_1}\sigma_{notNull(\vec{x}_1)}sql_1(\vec{z}_1)^D : & \\
        & \forall \vec{v} _1, \vec{v}_1' \in \pi_{\vec{y}_1}\sigma_{notNull(\vec{y}_1)}sql_1(\vec{z}_1)^D : \ldots : \forall \vec{v}_n, \vec{v}_n' \in \pi_{\vec{y}_n}\sigma_{notNull(\vec{y}_n)}sql_n(\vec{z}_n)^D : \\
        & (\vec{u},\vec{v}_1) \in \pi_{\vec{x}_1\vec{y}_1}\sigma_{notNull(\vec{x}_1,\vec{y}_1)}sql_1(\vec{z}_1)^D \land \ldots \land (\vec{v}_{n-1},\vec{v}_n) \in \pi_{\vec{x}_n\vec{y}_n}\sigma_{notNull(\vec{x}_n,\vec{y}_n)}sql_n(\vec{z}_n)^D \bigwedge \\
        & (\vec{u},\vec{v}_1') \in \pi_{\vec{x}_1\vec{y}_1}\sigma_{notNull(\vec{x}_1,\vec{y}_1)}sql_1(\vec{z}_1)^D \land \ldots \land (\vec{v}_{n-1}',\vec{v}_n') \in \pi_{\vec{x}_n\vec{y}_n}\sigma_{notNull(\vec{x}_n,\vec{y}_n)}sql_n(\vec{z}_n)^D \Rightarrow \vec{v}_1 = \vec{v}_1', \ldots, \vec{v}_n = \vec{v}_n'
      \end{align*}
      $\Updownarrow$ (Standard Translation AND assumptions on templates)\\
      \begin{align*}
        \forall \vec{u} \in \pi_{\vec{x}_1}\sigma_{notNull(\vec{x}_1)}sql_1(\vec{z}_1)^D : \\
        & \forall \vec{v}_1, \vec{v}_1' \in \pi_{\vec{y}_1}\sigma_{notNull(\vec{y}_1)}sql_1(\vec{z}_1)^D : \ldots : \forall \vec{v}_n, \vec{v}_n' \in \pi_{\vec{y}_n}\sigma_{notNull(\vec{y}_n)}sql_n(\vec{z}_n)^D : \\
        & (\vec{u}, \vec{v}_1, \ldots, \vec{v}_n) \in \pi_{\vec{x}_1\vec{y}_1\dots \vec{y}_n}(\sigma_{notNull(\vec{x}_1,\vec{y}_1)}\widehat{sql_1(\vec{z}_1)} \Join_{\vec{y}_1=\vec{x}_2} \sigma_{notNull(\vec{x}_2,\vec{y}_2)}sql_2(\vec{z}_2) \Join_{\vec{y}_2=\vec{x}_3}\dots  \Join_{\vec{y}_{n-1}=\vec{x}_n} \sigma_{notNull(\vec{x}_n,\vec{y}_n)}sql_n(\vec{z}_n))^D \bigwedge \\
        & (\vec{u}, \vec{v}_1', \ldots, \vec{v}_n') \in \pi_{\vec{x}_1\vec{y}_1\dots \vec{y}_n}(\sigma_{notNull(\vec{x}_1,\vec{y}_1)}\widehat{sql_1(\vec{z}_1)} \Join_{\vec{y}_1=\vec{x}_2} \sigma_{notNull(\vec{x}_2,\vec{y}_2)}sql_2(\vec{z}_2) \Join_{\vec{y}_2=\vec{x}_3}\dots  \Join_{\vec{y}_{n-1}=\vec{x}_n} \sigma_{notNull(\vec{x}_n,\vec{y}_n)}sql_n(\vec{z}_n))^D \Rightarrow \\ 
        &\hspace{11cm} \vec{v}_1 = \vec{v}_1', \ldots, \vec{v}_n = \vec{v}_n'\\
      \end{align*}
      $\Updownarrow$ (Definition of Functional Dependency)\\
      $ \vec{x}_1\rightarrow \vec{y}_1 \dots  \vec{y}_n$ is satisfied in $\pi_{\vec{x}_1\vec{y}_1\dots \vec{y}_n}(\widehat{sql_1(\vec{z}_1)} \Join_{\vec{y}_1=\vec{x}_2} sql_2(\vec{z}_2) \Join_{\vec{y}_2=\vec{x}_3}\dots  \Join_{\vec{y}_{n-1}=\vec{x}_n} sql_n(\vec{z}_n))^D$
      $\hfill\Box$
    \end{proof}
  \end{lemma}

  \begin{example}
    Consider the following set of $\T$-mappings for an OBDA setting $\O$:

\begin{lstlisting}[basicstyle=\ttfamily\scriptsize,mathescape,frame=none,frame=tb]
$\owlind{f(id,name)} \ \owl{P}_1 \ \owlind{g(friend)} \leftarrow$ $\SQL{\SELECT{} id, name, friend \FROM{} T}$ 
$\owlind{g(friend)} \ \owl{P}_2 \ \owlind{h(friend\_age)} \leftarrow$ $\SQL{\SELECT{} friend, friend\_age \FROM{} T}$
\end{lstlisting}
Then the lemma above suggests that the \vfd $f\mapsto^p{P_1P_2}$ is satisfied in $\O$
if and only if the \fd $\owl{\small id name}$ $\rightarrow$ \owl{\small friend friend\_age} is satisfied in $(T \Join_{friend = friend} T)^D$.
 \end{example}

 \begin{definition}[Optimizing Path \vfd]\label{def:optimizable-path}
   Let $t$ be a template, and $P_1,\dots, P_n$ be properties in $\T$. 
   An \emph{optimizing path \vfd} is an expression of the form
   $t \rightsquigarrow^p P_1 \cdots P_n$.
   An optimizing \vfd \emph{$t \rightsquigarrow^p P_1 \cdots P_n$ is satisfied in $\O$} if $t \mapsto^p P_1 \cdots P_n$ is satisfied in $\O$ and 
   \begin{equation}\label{eq:pathPrecondition}
     \pi_{\textbf{x}_{1}\textbf{y}_{1}\dots \textbf{y}_{n}}sql_1(\textbf{z}_1)^D \subseteq q_{path}^D        
   \end{equation}
   where
   \[q_{path} = \pi_{\textbf{x}_{1}\textbf{y}_{1}\dots \textbf{y}_{n}} ( sql_1(\textbf{z}_1)\Join_{\textbf{y}_{1}=\textbf{x}_{2}} sql_2(\textbf{z}_2) \Join_{\textbf{y}_{2}=\textbf{x}_{3}} \dots \Join_{\textbf{y}_{n-1}=\textbf{x}_{n}} sql_n(\textbf{z}_n) ).\]
\end{definition}

\begin{lemma}\label{def:opt2}
  Consider $n$ properties $P_1,\dots,P_n$ in $\T$ with $t_{r_i} = t_{d_{i+1}}$, for each $1 \le i < n$, and for which 
  ${t_d^1}\rightsquigarrow^p{P_1\cdots P_n}$ is satisfied in $\O$. 
   Then 
   \begin{align*}
     \pi_{\textbf{x}_{1}\textbf{y}_{1}\dots \textbf{y}_{n}}sql_1(\textbf{z}_1)^D & = q_{path}^D
   \end{align*}
   where $q_{path}$ is the same as in the Definition~\ref{def:optimizable-path}.
   \begin{proof}
     (Sketch) The argument is similar to the one of the proof for Lemma~\ref{def:opt1}, by using Lemma~\ref{lm:detect2}.
    \end{proof}
\end{lemma}

\begin{definition}[Optimizable path BGP]
  A BGP $\beta$ is optimizable w.r.t. $\mathfrak{v} = {t_d}\rightsquigarrow^p{P_1\dots P_n}$ 
  if \emph{(i)} $\mathfrak{v}$ is satisfied in $\O$;
\emph{(ii)} the BGP of triple patterns in $\beta$ involving properties  is of the form 
$\textowl{?v$_0$ ~P$_1$ ?v$_1$.~\dots ?v$_{n-1}$  P$_n$ ?v$_n$.}$;
and \emph{(iii)}  for every triple  pattern of the form $\textowl{?u  rdf:type  C}$ in $\beta$,
 $\textowl{?u}$ is the subject of some  
$P_i$ ($i=1\dots n$) 
and
$t_d^i \rightsquigarrow^d_{P_i} C$ is satisfied in $\O$
, or  $?u$ is  the object of some 
$P_i$ ($i=1\dots n$)  and $t_r^i \rightsquigarrow^r_{P_i} C$ is satisfied in $\O$.
\end{definition}

\paragraph{\textbf{Proofs for Main Results}}

~\\

\begin{theoremnum}{\ref{th:finalFD}}
  Let $\beta$ be an
  optimizable BGP w.r.t. ${t_d}\rightsquigarrow^x {P_1\dots P_n}$
  ($x=b,p$) in $\O$.  Let
  $\pi_{v/t_{d}^1,v_1/t_{r}^1,\ldots,v_n/t_{r}^n}sql_{\beta}$ be
  the SQL translation of $\beta$ as explained in
  Section~\ref{sec:sparql-qa}.  Let
  $sql_{\beta}'
  =sql_1(\vec{x}_{1},\vec{y}_{1}\dots,\vec{y}_{n}) $.
Then  $sql_{\beta}^D$ and $sql_{\beta}'^D$  return the same answers.
\end{theoremnum}

  \begin{proof}

    Assume that ${t_d}\rightsquigarrow ^p_{P_1\dots P_n}$. The proof for branching functional dependencies is analogous. 

    From the definition of  $\tra{}$ for triple pattern and the definition of the $\tra{} $ for $\Join$ for BGPs it follows that the BGP $\beta$ will be translated as:

    \begin{equation}\label{eq:1}
      \begin{array}{l}
        (\proj{v_0,v_1} \rho_{v/\mathtt{subj}}\,\rho_{v_1/\mathtt{obj}} \,\sigma_{\texttt{pred} = P_1}triple) \\
        ~~~~~~~~~~~~~~~~~~~~~~~~~~~~~~~~~~~~~\Join_{v_1=v_2} \\
        ~~~~~~~~~~~~~~~~~~~~~~~~~~~~~~~~~~~~~\vdots \\
        ~~~~~~~~~~~~~~~~~~~~~~~~~~~~~~~~~~~~~\Join_{v_{n-2}=v_{n-1}}\\
        (\proj{v_{n-1},v_n} \rho_{{v_{n-1}/\mathtt{subj}}}\,\rho_{\textit{v}_n/\mathtt{obj}} \,\sigma_{\texttt{pred} = P_n}triple) \\
      \end{array}
    \end{equation}

    The table triple is replaced by the definition of the triple patterns in the mappings as follows:

    \begin{equation}\label{eq:1}
      \begin{array}{l}
 (\proj{v_0,v_1} \rho_{v_0/\mathtt{subj}}\,\rho_{v_1/\mathtt{obj}} \,  \sigma_{\texttt{pred} = P_1} \pi_{\mathtt{pred}/P_1,\mathtt{subj}/t_d^{1},\mathtt{obj}/t_r^{1}}
(sql_1(\vec{z}_1))) \\
        ~~~~~~~~~~~~~~~~~~~~~~~~~~~~~~~~~~~~~\Join_{v_1=v_2} \\
        ~~~~~~~~~~~~~~~~~~~~~~~~~~~~~~~~~~~~~\vdots \\
        ~~~~~~~~~~~~~~~~~~~~~~~~~~~~~~~~~~~~~\Join_{v_{n-2}=v_{n-1}}\\
 (\proj{v_{n-1},v_n} \rho_{v_{n-1}/\mathtt{subj}}\,\rho_{v_n/\mathtt{obj}} \,  \sigma_{\texttt{pred} = P_n} \pi_{\mathtt{pred}/P_n,\mathtt{subj}/t_d^{n},\mathtt{obj}/t_r^{n}})(sql_n(\vec{z}_n)))) \\
      \end{array}
    \end{equation}

    %
    This expression can be simplified to:

\begin{equation}\label{eq:1}
      \begin{array}{l}
 (\proj{v_0,v_1} \rho_{v_0/\mathtt{subj}}\,\rho_{v_1/\mathtt{obj}} \,  \pi_{\mathtt{subj}/t_d^{1},\mathtt{obj}/t_r^{1}}
(sql_1(\vec{z}_1))) \\
        ~~~~~~~~~~~~~~~~~~~~~~~~~~~~~~~~~~~~~\Join_{v_1=v_2} \\
        ~~~~~~~~~~~~~~~~~~~~~~~~~~~~~~~~~~~~~\vdots \\
        ~~~~~~~~~~~~~~~~~~~~~~~~~~~~~~~~~~~~~\Join_{v_{n-2}=v_{n-1}}\\
 (\proj{v_{n-1},v_n} \rho_{v_{n-1}/\mathtt{subj}}\,\rho_{v_n/\mathtt{obj}} \, \pi_{\mathtt{subj}/t_d^{n},\mathtt{obj}/t_r^{n}})(sql_n(\vec{z}_n)))) \\
      \end{array}
    \end{equation}


By definition we know that the template in the range of $P_{i-1}$  coincide with the template in $P_i$.
Thus, we can remove them from the join over $u_i$'s in (\ref{eq:sql_tp}) and make the join over the attributes $\mathbf{x}_i, \mathbf{y}_i$ instead of the URIs. Therefore, $\beta$ can be rewritten to 
\begin{equation}\label{eq:2_proof}
	\begin{array}{l}
 \pi_{{v_0}/t_{d}^1,v_1/t_{r}^1,\ldots,v_n/t_{r}^n}
(sql_1(\vec{z}_1)\Join_{\vec{y}_{1}=\vec{x}_{2}} sql_2(\vec{z}_2) \Join_{\vec{y}_{2}=\vec{x}_{3}}\dots \Join_{\vec{y}_{n-1}=\vec{x}_{n}}  sql_n(\vec{z}_n))
\end{array}
\end{equation}

Since $\beta$ is optimizable, we know that
\begin{equation}\label{eq:contaiend}
\begin{array}{l}
\pi_{\vec{x}_{1}\vec{y}_{1}\dots \vec{y}_{n}}sql_1(\vec{z}_1) =  \pi_{\vec{x}_{1}\vec{y}_{1}\dots \vec{y}_{n}}
( sql_1(\vec{z}_1)\Join_{\vec{y}_{1}=\vec{x}_{2}}\ sql_2(\vec{z}_2) \Join_{\vec{y}_{2}=\vec{x}_{3}}\\
\hfill\dots \Join_{\vec{y}_{n-1}=\vec{x}_{n}} sql_n(\vec{z}_n) )
\end{array}
\end{equation}
Therefore, we can simplify (\ref{eq:2_proof}) to 
$\pi_{v/t_{d}^1,v_1/t_{r}^1,\ldots,v_n/t_{r}^n}
(sql_1(\mathtt{z}_1))$.
    This proves the Theorem.
    $\hfill\Box$
  \end{proof}


\subsection{Lifting Basic OBDA Instance Assumption}

We show that the ``basic OBDA instance assumption'' in Section~\ref{s:fds-in-obda} is not a real
restriction. A SPARQL query over a $\T$-mapping with predicates of
multiple templates can be rewritten to another SPARQL query over another
$\T$-mapping with predicates of only single template.

As usual, we assume an OBDA instance $((\T,\M,\Sigma),D)$, and 
let $\M_\T$ be a
$\T$-mapping.

Suppose a predicate $A$ is defined by $k$ mapping assertions using
  different template in $\M_\T$:
\[
\begin{aligned}
A(t_d^1(x), t_r^1(y)) & \leftarrow sql_1(z)  \\
\ldots \\
A(t_d^k(x), t_r^k(y)) & \leftarrow sql_k(z)  
\end{aligned}
\]

Define $\M_\T^A$ be the mapping obtained by replacing the assertions for the $A$
with the following $k$ mapping assertions defining $k$ fresh predicates
$A_i$ ($i=1,\ldots,k$):
\[
\begin{aligned}
A_1(t_d^1(x), t_r^1k(y)) & \leftarrow sql_1(z)  \\
\ldots \\
A_k(t_d^k(x), t_r^k(y)) & \leftarrow sql_k(z)  
\end{aligned}
\]
Suppose that $Q$ is a SPARQL query using predicate $A$. The idea is to
construct another SPARQL query $Q'$ such that $\llbracket Q \rrbracket _{(\M_\T,D)} = \llbracket Q' \rrbracket _{(\M_\T^A,D)}$. The construction is performed on each triple pattern using $A$. Suppose $B$ is a triple pattern occurring in $Q$; we take $B^+$
to be the union of
\[ B[A \mapsto A_i], \qquad i= 1\ldots k \]
where $B[A\mapsto A_i]$ is a triple pattern obtained by replacing all the occurrences of
$A$ in $B$ with $A_i$. Finally $Q'$ is defined as the SPARQL query obtained by
replacing all the triple patterns $B$ with $B^+$.

\begin{lemma}\label{lem:lift}
  $\llbracket B \rrbracket _{(\M_\T,D)} = \llbracket B^+ \rrbracket
  _{(\M_\T^A,D)}$
\end{lemma}

\begin{proof}
  We only prove the case where $B$ is a single triple pattern of the
  form $ B = (?x, \rdftype, A) $, since the case where $A$ is a property
  can be proved analogously. In this case,
  \[B^+ = (?x, \rdftype, A_1)~ \union \ldots \union~ (?x, \rdftype,
  A_k)\] 

  Suppose that $\{?x \mapsto a\} $ is a solution mapping, i.e., $A(a)$
  is in the RDF graph exposed by $\M_\T$ and $D$.  It follows that
  there is a mapping assertion
  $A(t(\vec{x})) \leftarrow sql_i(z) \in \M_\T$, such that
  $a=t(\vec{x}_0)$ for some template $t_0$ and tuple $x_0$. Since
  $A_i(t(\vec{x})) \leftarrow sql_i(z) \in \M_\T^A $, we have
  $A_i(t(\vec{x}_0))$ is in the RDF graph exposed by $\M_\T^A$ and
  $D$. Then $ \{ x \mapsto a \} $ is a solution mapping of
  $(?x, \rdftype, A_i)$ and also of $B^+$.

The other direction can be proved analogously.
\end{proof}

\begin{theorem}
\label{thm:lift}
$\llbracket Q \rrbracket _{(\M_\T,D)} = \llbracket Q' \rrbracket _{(\M_\T^A,D)}$
\end{theorem}

\begin{proof}
  The proof is a standard induction over the structure of the SPARQL
  queries. The base case of proof is the triple pattern case, and has been
  proved in Lemma~\ref{lem:lift}. The inductive case can be proved
  easily.
\end{proof}

By exhaustingly apply Theorem~\ref{thm:lift} to all predicates of
different templates, one can lift the restriction of ``basic OBDA
instance''.


\subsection{Wisconsin Benchmark}
\label{sec:wisconsin}

We setup an environment based on the Wisconsin Benchmark~\cite{wisconsin}.
This benchmark was designed for the systematic evaluation 
of data\-base performance with respect to different query characteristics.
The benchmark comes with a schema that is designed so one can quickly
understand the structure of each table and the distribution of each
attribute value.  This allows easy construction of queries that
isolate the features that need to be tested.  The benchmark also comes
with a data generator to populate the schema.  Unlike EPDS, the
benchmark database contains synthetic data that allows easily
specifying a wide range of retrieval queries.  For instance, in EPDS
it is very difficult to specify a selection query with a 20\% or 30\%
selectivity factor.  This task becomes even harder when we include
joins into the picture.

The benchmark defines a single table schema (which can be used to
instantiate multiple tables). The table, which we now call  ``Wisconsin table'', contains
16 attributes, and a primary key (\texttt{unique2})
with  integers from 0 to 100 million randomly ordered.

We refer the reader to~\cite{wisconsin} for details on the algorithm
that populates the schema.

\myPar{Dataset}
We used Postgres 9.1, and  DB2 9.7  as \ontop backends. 
The query optimizers were left with the default configurations.
All the table statistics  were updated. 

%
For each DB engine we created a database, each with 10 tables:  
5 Wisconsin tables
 ($\text{Tab}_i, i=1,\ldots,5$),  and 5 tables materializing the join of
the former tables. For instance, \texttt{\small view123} materializes the join 
of the tables \texttt{\small Tab1}, \texttt{\small Tab2}, and
\texttt{\small Tab3}.
Each table contains 100 million rows, and each of the databases
occupied ca. 400GB of disk space.


\myPar{Hardware}
We ran the experiments in an HP Proliant server with 24 Intel Xeon CPUs
(@3.47GHz), 106GB of RAM and five 1TB 15K RPM HD. 
\ontop was run  with 6GB Java heap space.
The OS is Ubuntu 12.04 LTS 64-bit edition.

\medskip In these experiments, we ran each query 3 times, and we
averaged the execution times.  There was a warm-up phase, where we ran
4 random queries not belonging to the tests.

\subsubsection{Evaluating the Impact of \vfd-based Optimization}

The experiments in this section measure the impact of optimization
based on \vfds.
Optimizations based on branching \vfds and path \vfds produce the same
effect in the resulting SQL query, therefore, for concreteness we
focus on branching \vfd.  The performance gain for path \vfd is
similar.

Recall that we started studying this scenario because EPDS contains
thousands of views that lack primary/foreign keys, and some of them
cannot be avoided in the mappings.  This prevents OBDA semantic
optimizations to take place.

The following experiments evaluate the trade-off of using views or
their definitions depending on:
\begin{inparaenum}[\it (i)]
	\item type of mappings (using views or view definitions);
	\item the complexity of the user query (\#~of SPARQL joins);
	\item the complexity of the mapping definition (\#~of SQL joins);
	\item the selectivity of the query;
	\item the \vfd optimization ON/OFF;
	\item the DB engine (DB2/PostgreSQL);
\end{inparaenum}

In the following we describe the queries, mappings and the OBDA
specifications and instances used in the different experiments.

\myPar{Queries}
In this experiment we tested a set of 36 queries each varying on:
\begin{inparaenum}[\it (i)]
\item the number of SPARQL joins (1-3), 
\item SQL joins in the mappings (1-4),  and 
\item selectivity of the query (3 different values). 
\end{inparaenum}
The SPARQL queries have the following shape:
\begin{center}
\begin{minipage}{0.75\textwidth}
\begin{lstlisting}[basicstyle=\ttfamily\scriptsize,mathescape,frame=none,frame=tb]
 SELECT ?x ?y  WHERE {
?x a ${:Class-n-SQLs}$ . ?x ${:Property_1-n-SQLs}$ ?y$_1$ . 
                     $\vdots$
?x ${:Property_m-n-SQLs}$ ?y$_m$ .  Filter( ?y$_m$ < k% ) }
\end{lstlisting}
\end{minipage}
\end{center}

where  \texttt{Class-n-SQLs} and \texttt{Property$_i$-n-SQLs} are classes and properties defined by mappings which source is either an SQL join of $n=1\dots4$ tables, or a materialized view of the join of $n$ tables. Subindex $m$ represents the number of SPARQL joins, 1 to 3.
Regarding the selectivity of $k\%$,	
we did the experiments with the following values:
\begin{inparaenum}[\it (i)]
\item $0.0001\%$  (100 results);
\item $0.01\%$ (10.000 results);
\item $0.1\%$ (100.000 results).
 \end{inparaenum}
These  queries  do not belong to the Wisconsin benchmark.

\myPar{OBDA Specifications} We have two OBDA settings, one where
classes and properties are populated using an SQL that use original
tables with primary keys (1-4 joins) ($K_1$); and a second one where
predicates are populated using materialized views (materializing 1-4
joins).  This second setting we tested with $\vfd$ optimization
($K_2$) and without optimization ($K_3$).  In the first OBDA setting,
all the property subjects are mapped into the tables primary keys.
There are no axioms in the ontology. All the individuals have the same
template $t$.

\ignore{
\myPar{\vfd Specifications}
Each  set of 36 queries was tested against 3 different scenarios. 
\begin{compactitem}
	\item : Using the mapping based on the original tables.  
	\item : Using the mapping based on materialized views  with  
	\item Using the mapping based on materialized views without
      optimization.
\end{compactitem}
}
 Let $S_t$ be the set of all individuals.
In $K_2$ there are  12 branching \vfds  of the 
form $S_t \mapsto^b{\textowl{:Property$_m$-n-SQLs}}$ for every $n=1\dots4$, $m=1\dots 3$. 
The optimizable \vfds contain intuitively the properties populated from the same view, that is,\\
\centerline{$S_t \mapsto^b{\textowl{:Property$_1$-n-SQLs},\textowl{:Property$_2$-n-SQLs},\textowl{:Property$_3$-n-SQLs}}$} \\
for $n=1\dots 4$.


\ignore{

\begin{example}
	
\begin{center}
\begin{minipage}{0.4\textwidth}
\begin{lstlisting}[basicstyle=\ttfamily\footnotesize,mathescape,frame=none,frame=tb]
 
SELECT ?x ?y  WHERE {
?x a  :2Tab1 . 
?x :Tab3unique2Tab3 ?y . 
Filter( ?y < 100 )
}
\end{lstlisting}
\end{minipage}
\end{center}
	
\end{example}
}

\myPar{Discussion and Results}
The  results of the experiments are shown in~ Figure~\ref{fig:FDResults}.
Each $q_{i/j}$ represents the query with $i$ SPARQL joins over properties
mapped to $j$ SQL joins. 


There is almost no difference between the results with
different selectivity, 
so for clarity we averaged the run times over different selectivities.
Since the experiment was run three times, each point in the figure
represents the average of 9 query executions. 

The experiment results  in Figure~\ref{fig:FDResults} show that 
all the SPARQL queries perform better in $K_2$ than in $K_3$ in both DB engines.
Moreover, in all cases queries in $K_2$  perform at least twice as fast as  the ones in $K_3$,
even getting close to the performance of $K_1$.

In \ontop-Postgres, the execution of the hardest SPARQL queries in
$K_1$  is 1 order of magnitude faster than in $K_2$. The execution of these queries in
$K_2$ is 4 times  faster than in $K_3$.
In \ontop-DB2, the performance gap between the SPARQL queries in 
$K_3$ and $K_2$ is smaller.
The SPARQL queries in
$K_1$  are slightly faster than the queries in $K_2$. The execution in
$K_2$ is 2 times  faster than in $K_3$.

In \ontop-Postgres and \ontop-DB2, the translations of the SPARQL queries resulting from the $K_3$ scenario 
contain self-joins of the non-indexed views that force 
the DB engines to create hash tables for all intermediate join results
which increases the start-up cost of the joins, and the overall
execution time.  One can observe that in both, \ontop-Postgres and
\ontop-DB2, the number of SPARQL joins strongly affect the performance
of the query in $K_3$.  In both cases, the SPARQL queries in $K_2$,
because of our optimization technique, get translated into a join-free
SQL query that requires a single sequential scan of the unindexed
view. However, the cost of scanning the whole view to perform a
non-indexed filter is still higher than the cost of joins (nested
joins in both) of the indexed tables in $K_1$.

 \begin{figure*}[t]
 \centering
 \pgfplotsset{every tick label/.append style={font=\tiny}}
\begin{minipage}{.75\textwidth}

		\begin{tikzpicture}
		\pgfplotsset{every axis/.append style={
			font=\footnotesize,
			semithick,
			tick style={thick}}}
		    \begin{axis}[
				title= {Postgres},
	    			height=4cm,
				width=9cm,
				xtick={1,2,3,4,5,6,7,8,9,10,11,12},
				xticklabels={$q_{1/1}$,$q_{1/2}$,$q_{1/3}$,$q_{2/1}$,$q_{2/2}$,
				$q_{2/3}$,$q_{3/1}$,$q_{3/2}$,$q_{3/3}$,$q_{4/1}$,$q_{4/2}$,$q_{4/3}$},
				ymin=0, 
				ymax=1000000,
                                scaled y ticks=false,
                                ytick={120000,240000,360000,480000,600000,720000,840000,960000},
                                yticklabels={2 m, 4 m, 6 m, 8 m, 10 m,12 m, 14 m, 16 m},
		        	grid=major,
				legend columns=-1
		    ]
\addplot[black,mark=triangle,  mark options={solid}] coordinates {
( 1,30470.6666666667)(2,31069.6666666667)(3,35456.3333333333)(4,28726)(5,31861)(6,37697.3333333333)(7,30490.3333333333)(8,33274)(9,38996.3333333333)(10,50795.3333333333)(11,45442.6666666667)(12,56197 )

};
 \label{fig:wis-mk-joins}

\addplot[blue,dash pattern=on 3pt off 6pt on 6pt off 6pt,mark=o,  mark options={solid}] coordinates {
( 1,167160.333333333)(2,257645)(3,348509)(4,248971)(5,379721)(6,504111.666666667)(7,315088)(8,472156)(9,629904.333333333)(10,498757.666666667)(11,742154)(12,988929.666666667 )
};
\label{fig:wis-mk-views_no_keys}


\addplot[red,densely dotted,mark=square,  mark options={solid}] coordinates {
( 1,72028.3333333333)(2,64178.6666666667)(3,69950)(4,127435.666666667)(5,124295)(6,126980)(7,161180)(8,161647.666666667)(9,165631.666666667)(10,260033.333333333)(11,259428.333333333)(12,261718.333333333 )
};
\label{fig:wis-mk-views_tuned}

\end{axis}
\end{tikzpicture}

\end{minipage}

\begin{minipage}{.75\textwidth}
	
		\begin{tikzpicture}
		\pgfplotsset{every axis/.append style={
			font=\footnotesize,
			semithick,
			tick style={thick}}}
		    \begin{axis}[
				title= {DB2},
	    			height=4cm,
				width=9cm,
				xtick={1,2,3,4,5,6,7,8,9,10,11,12},
				xticklabels={$q_{1/1}$,$q_{1/2}$,$q_{1/3}$,$q_{2/1}$,$q_{2/2}$,
				$q_{2/3}$,$q_{3/1}$,$q_{3/2}$,$q_{3/3}$,$q_{4/1}$,$q_{4/2}$,$q_{4/3}$},
				ymin=0, 
				ymax=1000000,
                                scaled y ticks=false,
                                ytick={120000,240000,360000,480000,600000,720000,840000,960000},
                                yticklabels={2 m, 4 m, 6 m, 8 m, 10 m,
                                12 m, 14 m, 16 m},
		        	grid=major,
				legend columns=1
		    ]
\addplot[black,mark=triangle,  mark options={solid}] coordinates {

( 1,32063)(2,32025.6666666667)(3,31858.3333333333)(4,37669)(5,37811.6666666667)(6,38093.6666666667)(7,38884.6666666667)(8,37764.6666666667)(9,37449.6666666667)(10,42848.3333333333)(11,43310.6666666667)(12,43756.6666666667 )
};
 \label{fig:wis-mk-db2-joins}


\addplot[blue,dash pattern=on 3pt off 6pt on 6pt off 6pt,mark=o,  mark options={solid}] coordinates {
( 1,108271.333333333)(2,187635.333333333)(3,207350.666666667)(4,143720.666666667)(5,205474)(6,285039.333333333)(7,202681.333333333)(8,293285)(9,410668.333333333)(10,249598)(11,352131.333333333)(12,473457.333333333 )
};
\label{fig:wis-mk-db2-views_no_keys}


\addplot[red,densely dotted,mark=square,  mark options={solid}] coordinates {
( 1,48713.6666666667)(2,48536.3333333333)(3,48388.3333333333)(4,66582.6666666667)(5,66504.6666666667)(6,67014.6666666667)(7,97762)(8,97610.3333333333)(9,97716.6666666667)(10,118405.666666667)(11,118746.666666667)(12,118494.666666667 )
};
\label{fig:wis-mk-db2-views_tuned}

\end{axis}
\end{tikzpicture}
\end{minipage}

\begin{minipage}{.90\textwidth}
	\centering
		\footnotesize{
			\ref{fig:wis-mk-joins} $K_1$ (Querying the tables directly)\\
			\ref{fig:wis-mk-views_tuned} $K_2$ (Views with optimization enabled)\\
			\ref{fig:wis-mk-views_no_keys} $K_3$ (Views no optimization)~~~
			}
\end{minipage}
\caption{Experiments showing the impact of the optimization technique based on \vfd}
\label{fig:FDResults}	
\end{figure*}
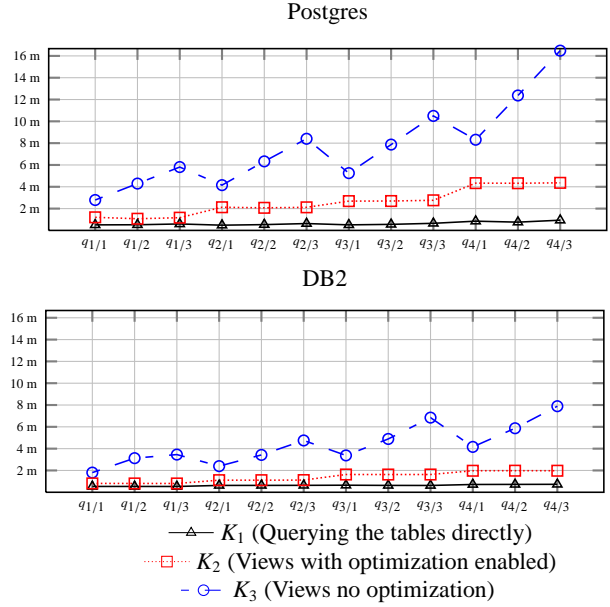


\ignore{
\paragraph{Observations in Experiments with Views}

\mysql implements only the \emph{nested loop} algorithm for realizing the join operation. A different approach is the \emph{Hash Join} algorithm, implemented both in \dbtwo and \postgres which, by exploiting a combination of analysis of the physical data and usage of the main memory, often significantly outperforms the \emph{nested loop} procedure. In particular, the \emph{Hash Join} allows for a fast join on columns which are not explicitly defined as keys, but that are keys in practice; this because an index for those columns is constructed on-the-fly and in the main memory, and therefore the only overhead is the phase of the index construction (one single table scan). Below we show the execution plan of \postgres:
}

\subsubsection{Evaluating the Impact of Exact Mappings}\label{sec:expr-exactMap}




In this test we evaluate the exact mapping optimization technique described in Section~\ref{sec:exact-mappings-obda}.
This experiment is inspired by the use case in EPDS where optimization
based on exact mapping can help.
The following experiments evaluate the impact of the optimization
depending on:
%
\begin{inparaenum}[(i)]
\item the complexity of the query (\# of SPARQL joins);
\item the selectivity of the query;
\item the number of specified  exact classes;
\item the DB engine (DB2/PostgreSQL). 
\end{inparaenum}

In the following we describe the tables, ontology, mappings, queries
and exact predicate specifications used in the experiment.

\myPar{OBDA Specifications}
The ontology contains four classes
$A_1, A_2, A_3, A_4$, one object property $R$ and one data property
$S$.  The classes form a hierarchy
\[A_1 ~\rdfssubclassof~ A_2, ~ A_2 ~\rdfssubclassof~ A_3, ~A_3 ~\rdfssubclassof~ A_4.\]
The mappings for classes $A_i$ ($i=1,\ldots,4$) are
defined over the primary key of Tab$_i$ with different
filters, in such a way that each $A_i$ is exact. The mappings for $R$ and
$S$ are defined over the primary key column and another unique column (\texttt{unique1})
of Tab$_5$.

\myPar{Queries}
 In this experiment we tested 6
queries ($q_1,\ldots,q_6$) 
varying on: (i) the number of classes and properties in the SPARQL (1-3) and (ii) the
classes used in the query.
For instance, $q_3$ is 
\begin{center}
\begin{minipage}{0.65\textwidth}
\begin{lstlisting}[basicstyle=\ttfamily\scriptsize,mathescape,frame=none,frame=tb]
SELECT * WHERE {?x a :A3. ?x :R y. ?y a :A4. 
  OPTIONAL { ?x :S ?u . } OPTIONAL { ?y :S ?v . }. }
\end{lstlisting}
\end{minipage}
\end{center}

\myPar{Exact Concepts}
We consider the following four exact concept specifications:
$E_0 = \emptyset$, 
$E_1 = \{A_1,A_2\}$, 
	$E_2 = \{A_1,A_2, A_3\}$,
$E_3 = \{A_1,A_2, A_3, A_4\}$.
Observe that $E_0$ corresponds to the case where no exact mapping optimization is applied.

 \begin{figure}[t]
   \centering
 \pgfplotsset{every tick label/.append style={font=\tiny}}
\begin{minipage}{.75\textwidth}
  \begin{tikzpicture}
    \pgfplotsset{every axis/.append style={
	font=\footnotesize,
	semithick,
	tick style={thick}}}
    \begin{axis}[
	title= {Postgres},
	height=4cm,
	width=9cm,
	xtick={1,2,3,4,5,6,7,8,9,10,11,12},
	xticklabels={$q_{1}$,$q_{2}$,$q_{3}$,$q_{4}$,$q_{5}$,
	  $q_{6}$},
	ymin=0, 
	ymax=1200000,
        scaled y ticks=false,
        ytick={120000,240000,360000,480000,600000,720000,840000,960000,1080000,1200000},
        yticklabels={2 m, 4 m, 6 m, 8 m, 10 m,12 m, 14 m, 16 m, 18 m,
          20 m},
	grid=major,
	legend columns=1
      ]
      \addplot[blue,mark=o,  mark options={solid}] coordinates {
          (1,150680) (2,146562) (3,201718) (4,246921) (5,1200000) (6,1200000)
      };
      \label{e0}

      \addplot[black,dash pattern=on 3pt off 6pt on 6pt off 6pt,mark=triangle,  mark options={solid}] coordinates {
        (1,15725) (2,45653) (3,93230) (4,140780) (5,286021) (6,434190)
      };
      \label{e1}

      \addplot[red,densely dotted,mark=square,  mark options={solid}] coordinates {
        (1,15803) (2,30862) (3,32814) (4,63711) (5,65933) (6,128035)
      };
      \label{e2}
      \addplot[green,solid,mark=x,  mark options={solid}] coordinates {
        (1,15514) (2,16002) (3,17408) (4,17809) (5,19406) (6,19711)
      };
      \label{e3}
    \end{axis}
  \end{tikzpicture}
\end{minipage}
\begin{minipage}{.75\textwidth}
		\begin{tikzpicture}
		\pgfplotsset{every axis/.append style={
			font=\footnotesize,
			semithick,
			tick style={thick}}}
		    \begin{axis}[
				title= {DB2},
                height=4cm,
				width=9cm,
				xtick={1,2,3,4,5,6,7,8,9,10,11,12},
	xticklabels={$q_{1}$,$q_{2}$,$q_{3}$,$q_{4}$,$q_{5}$,
	  $q_{6}$},
				ymin=1000, 
				ymax=1200000,
                                scaled y ticks=false,
                                ytick={120000,240000,360000,480000,600000,720000,840000,960000,1080000,1200000},
                                yticklabels={2 m, 4 m, 6 m, 8 m, 10 m,12 m, 14 m, 16 m, 18 m,
                                  20 m},
		        	grid=major,
                    legend columns=1 ]
      \addplot[blue,mark=o,  mark options={solid}] coordinates {
        (1,196456) (2,391016) (3,1146362) (4,1200000) (5,1200000) (6,1200000)
      };
      \label{fig:wis-ex-db2-e0}

      \addplot[black,dash pattern=on 3pt off 6pt on 6pt off 6pt,mark=triangle,  mark options={solid}] coordinates {
        (1,97807) (2,294653) (3,571352) (4,845108) (5,1200000) (6,1200000)

      };
      \label{fig:wis-ex-db2-e1}

      \addplot[red,densely dotted,mark=square,  mark options={solid}] coordinates {
        (1,96094) (2,193611) (3,188214) (4,360226) (5,1200000) (6,1200000)
      };
      \label{fig:wis-ex-db2-e2}
      \addplot[green,solid,mark=x,  mark options={solid}] coordinates {
        (1,97382) (2,95404) (3,104967) (4,96599) (5,931600) (6,1106524)
      };
      \label{fig:wis-ex-db2-e3}
\end{axis}
\end{tikzpicture}
\end{minipage}
\begin{minipage}{.90\textwidth}
	\centering
		\footnotesize{
			\ref{e0} $E_0$ (No exact mappings)~~~ 
			\ref{e1} $E_1$ (1 exact mapping)~~~\\
			\ref{e2} $E_2$ (2 exact mappings)~~~
			\ref{e3} $E_2$ (3 exact mappings)~~~
			}
\end{minipage}
\caption{Experiments showing the impact of the optimization technique based on Exact Mappings}
\label{fig:exact-map-left}	
\end{figure}
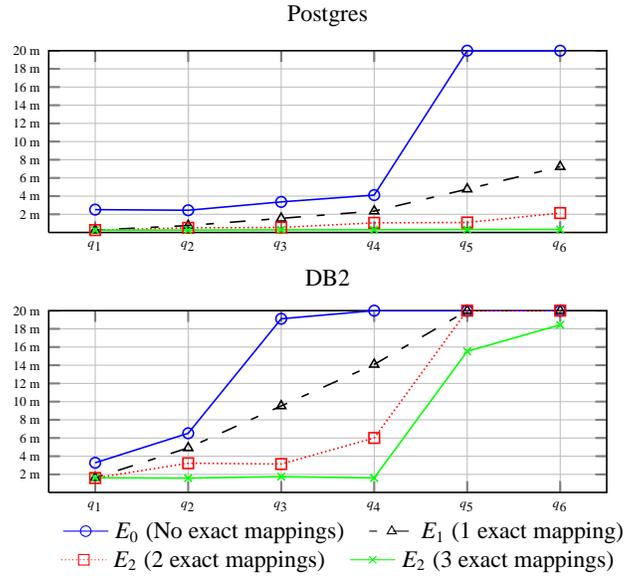

\myPar{Discussion and Results} 
The  results of the experiments are shown
in Figure~\ref{fig:exact-map-left}. 
The results show that the exact mapping optimization improves the performance of
all SPARQL queries in both database engines. In particular, under the full optimization
setting $E_3$, none of the queries time out (20 mins), and the hardest queries  perform orders of magnitude faster than in $E_0$ and even $E_1$.

The performance gain is the result of the elimination of redundant unions.  
For instance, under $E_0$, SPARQL query $q_3$ is translated into a SQL
query with 12 unions, but 11 of them are redundant; applying $E_3$ removes
all the redundant unions.

\newpage

\subsection{Experiments Material and Tools}

All the material related to the Wisconsin experiment, as well as the tools used to find exact mappings and virtual functional dependencies, can be found on 

\begin{center}
  \url{https://github.com/ontop/ontop-examples/tree/master/ruleml-2016}.
\end{center}



\fi

\end{document}